\documentclass[12pt]{article}
\usepackage{amsthm}
\usepackage{amsmath}
\usepackage{amsfonts}
\usepackage{graphicx}
\usepackage{latexsym}
\usepackage{amssymb}

\DeclareMathAlphabet{\mathbfsl}{OT1}{ppl}{b}{it} 

\newcommand{\Strut}[2]{\rule[-#2]{0cm}{#1}}

\newcommand{\deff}{\mbox{$\stackrel{\rm def}{=}$}}

\makeatletter
 \DeclareRobustCommand{\nsbinom}{\genfrac[]\z@{}}
 \makeatother
 
 \newcommand{\sbinomq}[2]{\nsbinom{{#1}}{{#2}}_{q}}
  \newcommand{\sbinomtwo}[2]{\nsbinom{{#1}}{{#2}}_{2}}

\newcommand{\field}[1]{\mathbb{#1}}

\newcommand{\F}{\field{F}}

\newcommand{\dS}{\field{S}}
\newcommand{\T}{\field{T}}
\newcommand{\V}{\field{V}}
\newcommand{\cA}{{\cal A}}

\newcommand{\cC}{{\cal C}}

\newcommand{\cF}{{\cal F}}

\newcommand{\cS}{{\cal S}}

\newcommand{\cP}{{\cal P}}

\newcommand{\cU}{{\cal U}}
\newcommand{\cV}{{\cal V}}

\newcommand{\cN}{{\cal N}}

\newcommand{\cG}{{\cal G}}

\newcommand{\B}{{\mathbb B}}
\newcommand{\C}{{\mathbb C}}
\newcommand{\CMRD}{\C^{\textmd{MRD}}}

\newcommand{\linadd}{\kern1pt\mbox{\small$\boxplus$}\kern1pt}

\newtheorem{defn}{Definition}
\newtheorem{theorem}{Theorem}
\newtheorem{lemma}{Lemma}

\newtheorem{cor}{Corollary}

\newtheorem{example}{Example}
\newtheorem{problem}{Research problem}

\begin{document}

\bibliographystyle{plain}

\title{
\begin{center}
Problems on $q$-Analogs in Coding Theory
\end{center}
}
\author{
{\sc Tuvi Etzion}\thanks{Department of Computer Science, Technion,
Haifa 32000, Israel, e-mail: {\tt etzion@cs.technion.ac.il}.}}

\maketitle

\begin{abstract}
The interest in $q$-analogs of codes and designs has been increased
in the last few years as a consequence of their new application
in error-correction for random network coding.
There are many interesting theoretical, algebraic,
and combinatorial coding problems
concerning these $q$-analogs which remained unsolved.
The first goal of this paper is to make a short summary of
the large amount of research which was done in the area
mainly in the last few years and to provide
most of the relevant references. The second goal of this paper is to present
one hundred open questions and problems for future research,
whose solution will advance the knowledge in this area. The third goal of
this paper is to present and start
some directions in solving some of these problems.
\end{abstract}

\vspace{0.5cm}


\footnotetext[1] { This research was supported in part by the Israeli
Science Foundation (ISF), Jerusalem, Israel, under
Grant 10/12.}

\newpage
\section{Introduction}

Let $\F_q$ be the finite field with $q$ elements
and let $\F_q^n$ be the set of all vectors
of length $n$ over $\F_q$. $\F_q^n$ is a vector
space with dimension $n$ over $\F_q$. The \emph{projective space} $\cP_q(n)$,
is the set of all subspaces of $\F_q^n$, including $\{ {\bf 0} \}$
and $\F_q^n$. For a given integer $k$,
$1 \leq k \leq n$, let $\cG_q(n,k)$ denote the set of all
$k$-dimensional subspaces of $\F_q^n$. $\cG_q(n,k)$ is often
referred to as Grassmannian. It is well known that
$$ \begin{small}
| \cG_q (n,k) | = \sbinomq{n}{k}
\deff \frac{(q^n-1)(q^{n-1}-1) \cdots
(q^{n-k+1}-1)}{(q^k-1)(q^{k-1}-1) \cdots (q-1)}
\end{small}
$$
where $\sbinomq{n}{k}$ is the $q-$\emph{ary Gaussian
coefficient}~\cite[pp. 325-332]{vLWi92}.

A \emph{subspace code} $\C$ is a subset of $\cP_q(n)$ and
a \emph{Grassmannian code} (known also as a \emph{constant dimension code})
$\C$ is a subset of
$\cG_q(n,k)$. Clearly, the Grassmannian codes are subset
of the subspace codes. In recent years there has been an increasing
interest in subspace codes as a result of their
application to error-correction in random network coding as was
demonstrated by Koetter and Kschischang~\cite{KoKs08}. But, the
interest in these codes has been also before this application,
since Grassmannian codes are $q$-analogs of the well
studied constant weight codes~\cite{BSSS}.
For example, the nonexistence of nontrivial perfect codes
in the Grassmann scheme was proved in~\cite{Chi87,MaZh95}. The
well-known concept of $q$-analogs replaces subsets by subspaces of
a vector space over a finite field and their orders by the
dimensions of the subspaces. In particular, the $q$-analog
of a constant weight code in the Johnson space is a constant
dimension code in the Grassmannian space. $q$-analogs of various
combinatorial objects are well known~\cite[pp. 325-332]{vLWi92}.
Of special interest are $q$-analogs in extremal
combinatoric as well as other well-known combinatorial
problems, e.g.~\cite{BBS12,BBSW,CGR06,ChPa10,FrWi86,Hsi75}.
The related techniques and results might be of usage in coding theory.

It turns out that the natural measure of distance in $\cP_q(n)$ is given by
$$
d_S (X,Y)  \deff \dim X+ \dim Y -2 \dim\bigl( X\, {\cap} Y\bigr)~,
$$
for all $X,Y \in \cP_q(n)$. This measure of distance is called the
\emph{subspace distance} and it is the $q$-analog of the
Hamming distance in the Hamming space.

From the point of view of error-correction in random network coding
it is better to use another measure of distance,
called the \emph{injection distance} given by
$$
d_I(X,Y) = \max \{ \dim(X),\dim(Y) \} - \dim ( X \cap Y)~.
$$

The injection distance is the $q$-analog of the asymmetric
distance between binary words~\cite{Etz91,EtOs98,Shi02}.
Both, the subspace distance and the injection distance are metrics.
When $X$ and $Y$ have the same dimension $k$, the subspace metric and
the injection metric coincide. If $X,~Y \in \cG_q(n,k)$ then
we can define their distance slightly different as follows,

$$
d_G(X,Y) = k - \dim ( X \cap Y)~.
$$
This measure of distance
will be called the \emph{Grassmannian distance}.
The Grassmannian distance is the $q$-analog of the
Johnson distance used for constant weight codes.
It is equal to half of the subspace distance and it is equal exactly
to the injection distance, i.e. $2 d_G (X,Y) = d_S(X,Y)=2d_I(X,Y)$.

The three measures of distance are metrics and they have related
families of graphs $G(\cG_q(n,k))$, $G(\cP_q^S(n))$, and $G(\cP_q^I(n))$,
for the Grassmannian metric, the subspace metric, and the injection metric,
respectively. The vertices of the graphs are the subspaces of $\cG_q(n,k)$,
$\cP_q(n)$, and $\cP_q(n)$, respectively. Two vertices
$X$ and $Y$ are connected by an edge
if the distance between $X$ and $Y$ is \emph{one} in the Grassmannian
metric, the subspace metric, and the injection metric, respectively.

This paper has three goals. The first one is to give a brief survey
on the known results on
$q$-analog problems in coding theory and
related problems in block design. This
will enable to put in one
place all the relevant references in this area.
The second goal of this paper is to suggest
problems and questions for future research
in this area and to motivate further research
on the related topics. One hundred such questions and problems are presented.
We remark that some problems might be contained in other
problem, but the level of difficulty of such contained
problems should be different. Other problems for future
research will be understood from the text and the context.
The third goal is to start some directions
of solution in some of the given research
problems. Each one of the sections which follows
will be devoted to another topic in this area.
In Section~\ref{sec:cdc} we will discuss bounds on the
size of Grassmannian codes. We will start to construct
a table on the lower and upper bounds on the size of constant dimension
codes in $\cP_2(n)$ for $n \leq 7$. In Section~\ref{sec:multiL}
we will discuss the multilevel construction which is the most
effective and simple way to construct very large Grassmannian
codes or subspace codes with a given distance measure.
The Grassmannian codes are usually larger than any other
known codes with the same parameters.
Bounds on the size of
subspace codes with the subspace distance, are discussed
in Section~\ref{sec:subspace}. We will present lower and upper bounds
on the size of the codes in $\cP_2(n)$ for $n \leq 7$.
We will also present a new general bound and a related
interesting problem, which is also interesting in
the context of extremal combinatorics.
In Section~\ref{sec:injection} we will discuss bounds on the size of
subspace codes with the injection distance.
A new interesting cyclic code which we present will
raise some interesting questions.
In Section~\ref{sec:Steiner} the existence
question for $q$-analogs of Steiner systems which are one family of
optimal Grassmannian codes, are discussed.
We will present a new method which might lead for
exclusion of parameters in which $q$-analog of Steiner systems can exist.
In Section~\ref{sec:spreads} we will discuss
the concepts of spreads and partial spreads which are used in
projective geometry, but they are also optimal Grassmannian codes.
Section~\ref{sec:rank} will be devoted to rank-metric codes which
are highly connected to Grassmannian codes and also to subspace codes.
Encoding and decoding of subspace codes are discussed in Section~\ref{sec:decoding}.
Designs over~$F_q$, i.e. $q$-analogs of designs
are considered in Section~\ref{sec:designs}.
In Section~\ref{sec:covering} we consider covering proplems
in the projective space and the Grassmannian space.
The asymptotic behavior of codes and designs in
the projective space and the Grassmannian space
is discussed in Section~\ref{sec:asymptotic}.
Disjoint spreads are considered in Section~\ref{sec:parallel}.
In Section~\ref{sec:other} we discuss three more $q$-analog of coding problems.
The first one is $q$-analog of Gray codes and in particular the
$q$-analog of the the middle levels
problem. The second problem is the existence question
of complements in the projective space.
The third problem is the existence question of linear codes in the projective space.

\section{Constant Dimension Codes}
\label{sec:cdc}

A Grassmannian code is also called a constant dimension code
since all the codewords have the same dimension. An $(n,\delta,k)_q$ code is
a subset of $\cG_q(n,k)$ with minimum Grassmannian distance $\delta$.
Let $\cA_q(n,\delta,k)$ denote the maximum size of an $(n,\delta,k)_q$ code.
Koetter and Kschischang~\cite{KoKs08}, Etzion and Vardy~\cite{EtVa11}
developed several upper bounds on $\cA_q(n,\delta,k)$.
For a subspace code $\C$ we define the \emph{orthogonal complement}
$\C^\perp$ as the code which consists of the dual subspaces of $\C$,
i.e. $\C^\perp \deff \{ X^\perp ~:~ X \in \C \}$. $\C$ and $\C^\perp$
have the same minimum distance (subspace, Grassmannian, or injection).
Therefore, $\cA_q(n,\delta,k)=\cA_q(n,\delta,n-k)$ and hence
in the sequel we will
also consider only $(n,\delta ,k)_q$ codes and only bounds
on $\cA_q(n,\delta,k)$
for which $2k \leq n$. The upper bounds on $\cA_q(n,\delta,k)$
are usually the $q$-analog of the bounds in the
Hamming scheme and the Johnson scheme.
These include the sphere packing bound and the
Singleton bound~\cite{KoKs08}, the Johnson bounds~\cite{EtVa11,XiFu09} from
which the most important one is:

\begin{theorem}
\label{thm:Joh}
$
\cA_q(n,\delta,k) \leq \left\lfloor \frac{q^n-1}{q^k-1} \cA_q (n-1,\delta,k-1)
\right\rfloor ~.
$
\end{theorem}

Theorem~\ref{thm:Joh} can be iterated to obtain the iterated Johnson bound
and the packing bound, which was proved earlier also in~\cite{WXSN03}, where
the context was linear authentication codes.

\begin{theorem}
\label{thm:packing}
$$
\cA_q(n, \delta,k) \leq \left\lfloor
\frac{q^n-1}{q^k-1} \left\lfloor \frac{q^{n-1}-1}{q^{k-1}-1}
\cdots \left\lfloor \frac{q^{n+\delta -k}-1}{q^{\delta}-1}
\right\rfloor \cdots \right\rfloor \right\rfloor \leq
\frac{\sbinomq{n}{k-\delta+1}}{\sbinomq{k}{k-\delta+1}}~.
$$
\end{theorem}

\vspace{0.3cm}

As for lower bounds on $\cA_q(n,\delta,k)$,
in~\cite{KoKs08} there is a construction,
of codes based on linearlized polynomials,
which yields the bound $\cA_q(n,\delta,k) \geq q^{(n-k)(k-\delta+1)}$.
The same bound was developed in~\cite{SKK08} by using lifted
rank-metric codes. In this context we define the rank distance
and rank-metric codes which play an important role in the
discussion on subspace codes.
For two $k \times \ell$ matrices $A$ and $B$ over $\F_q$ the {\it
rank distance} is defined by
$$
d_R (A,B) \deff \text{rank}(A-B)~.
$$
A $[k \times \ell,\varrho,\delta]$ {\it rank-metric code} $\cC$
is a linear code, whose codewords are $k \times \ell$ matrices
over $\F_q$; they form a linear subspace with dimension $\varrho$
of $\F_q^{k \times \ell}$, and for each two distinct codewords $A$
and $B$ we have that $d_R (A,B) \geq \delta$. For a
${[k \times \ell,\varrho,\delta]}$ rank-metric code $\cC$ it was proved
in~\cite{Del78,Gab85,Rot91} that
\begin{equation}
\label{eq:MRD} \varrho \leq
\text{min}\{k(\ell-\delta+1),\ell(k-\delta+1)\}~.
\end{equation}
This bound, called the Singleton bound for rank-metric codes,
is attained for all
feasible parameters. The codes which attain this bound are called {\it
maximum rank distance} codes (or MRD codes in short).

Let $A$ be a $k \times \ell$ matrix over $\F_q$ and let $I_k$ be a
$k \times k$ identity matrix. The matrix $[ I_k ~ A ]$ can be
viewed as a generator matrix of a $k$-dimensional subspace of
$\F_q^{k+\ell}$, and it is called the \emph{lifting} of $A$~\cite{SKK08}.
When all the codewords of a rank-metric code
$\cC$ are lifted to $k$-dimensional subspaces, the result is
a constant dimension code $\C$. If $\cC$ is an MRD code
then $\C$ is called a \emph{lifted MRD
code}~\cite{SKK08}. This code will be denoted by $\CMRD$.
This code is not maximal and it can be extended by using
a multilevel construction~\cite{EtSi09} as described in the next section.
An upper bound on the size of a code which contains $\CMRD$
can be found in~\cite{EtSi13}. Codes based on linearlized
polynomials, where each code contains the related code based
on linearlized polynomial constructed in~\cite{KoKs08}, were
developed in~\cite{Ska10}. But, these codes are smaller in size
then the codes obtained by the multilevel construction.

Another family of Grassmannian codes are codes which admit
a certain automorphism group. These kind of codes are discussed
in~\cite{EtVa11,KoKu08}. One of the most interesting family of
such codes are the cyclic codes.
Let $\F_{q^n}$ be a finite field with $q^n$ elements, where $q$ is
a power of a prime number,
and let $\alpha$ be a primitive element in $\F_{q^n}$.
It is well-known that there is an isomorphism between $\F_{q^n}$ and
$\F_q^n$, where the \emph{zero} elements are mapped into each other,
and $\alpha^i \in \F_{q^n}$, $0 \leq i \leq q^n-2$, is mapped into its
$q$-ary $n$-tuple representation in $\F_q^n$, and vice versa.
Using this mapping, a subspace
of~$\F_q^n$ is represented by the corresponding elements of~$\F_{q^n}$.
Let $\alpha$ be a primitive element of $\F_{q^n}$. We say
that a code $\C \subseteq \cP_q(n)$ is \emph{cyclic} if it has the
following property: whenever
$\{ {\bf 0} ,\alpha^{i_1},\alpha^{i_2},\ldots,\alpha^{i_m}\}$ is a
codeword of $\C$, so is its cyclic shift $\{ {\bf 0},
\alpha^{i_1+1},\alpha^{i_2+1},\ldots,\alpha^{i_m+1} \}$. In other
words, if we map each vector space $V \,{\in}\, \C$ into the
corresponding binary characteristic vector $x_V =
(x_0,x_1,\ldots,x_{2^n-2})$ given by
$$
x_i = 1 ~~\text{if $\alpha^i {\in}\kern1pt V$}
\hspace{4ex}\text{and}\hspace{4ex} x_i = 0 ~~\text{if $\alpha^i
{\not\in}\, V$}
$$
then the set of all such characteristic vectors is closed under
cyclic shifts. We note that the property of being cyclic does
\emph{not} depend on the choice of a primitive element $\alpha$ in
$\F_{q^n}$.

Cyclic codes have a nice automorphism group. But, there are other automorphisms
which can be forced on the code. One example is the use of the Frobenius mapping
which was used for example in~\cite{BEOVW}.
Some automorphisms of constant dimension codes were
studied in~\cite{Tra13}. Constructions for small dimensions
might be attractive in this context. Interesting codes admitting some automorphisms
were constructed in~\cite{BrRe12}. Some of these codes have an interesting
combinatorial structure and some were found only by computer search.
These were used to obtain lower bounds on $\cA_2(n,2,3)$. Lower
bounds on $\cA_q(n,2,3)$ were also considered in~\cite{EtSi13}.
Codes with subspaces of dimension 3 are of special
interest mainly since the value of
$\cA_q(n,\delta ,2)$ is known for all parameters.

Another family of codes which was considered, even so the codes were
not as large as in previous constructions, are the
orbit codes~\cite{MTR11,TMBR11,TMR11,TrRo13}. This family
of codes might deserve further attention in the future.
Another line of research for Grassmannain codes is based on
Schubert calculus and Pl\"{u}cker coordinates. These were
considered for example in~\cite{Gha13,TSR13,TrRo12} and their
further research might lead to new interesting results.
Lexicodes in the Grassmannian and their search
were discussed in~\cite{SiEt11a}.
Codes which are able to correct also
errors in coordinates (such as deletion or localized errors)
are considered in~\cite{Cai13,SMN12}.
A question on the size of equidistant Grassmannian codes
was asked in~\cite{EtVa11}. In an equidistant
code, any two codewords have the same distance,
which is clearly the minimum distance of the code.
This problem is highly connected to problems in
extremal combinatorics, e.g.~\cite{FrWi86,Hal77}.
Some work in this direction was
done lately in~\cite{Han12,Han12a,Han13}.
Finally, there are other related coding problems in the Grassmannian.
For example, the intersection size of balls around codewords has
some interesting applications~\cite{YSLB}, where the case
for the intersection size two balls in the Grassmannian was considered.

We conclude this section with our first list of research problems.

\begin{problem}
Find a systematic construction for cyclic codes in $\cG_q(n,k)$.
\end{problem}

\begin{problem}
Find new methods to construct large constant constant dimension
codes which are not based on lifting of rank-metric codes.
\end{problem}

\begin{problem}
Can we derive some constraints to form a linear programming
method to obtain new upper bounds on $\cA_q(n,\delta,k)$?
\end{problem}

\begin{problem}
Find new general upper bounds on the size of Grassmannian codes.
\end{problem}

\begin{problem}
Find a general upper bound on the size of a Grassmannian code
which contains $\CMRD$.
\end{problem}

\begin{problem}
Improve the lower and upper bounds on $\cA_q(n,\delta,3)$ for
specific values of $n$ and for large $n$.
\end{problem}

\begin{problem}
What is the size of the largest equidistant $(n,\delta,k)_q$ code?
\end{problem}

\begin{problem}
Find new applications for Grassmannian codes.
\end{problem}

\begin{problem}
Find the size of the intersection of more than
two balls in the Grassmannian.
\end{problem}

\begin{problem}
\label{pr:tab_cdc}
Compile a table for the lower and upper bounds on $\cA_q(n,\delta,k)$
for small values of $q$, $n$, $\delta$, and $k$.
\end{problem}

We start with some tables related to Research problem~\ref{pr:tab_cdc}.
We consider the first few tables for $q=2$. Each table will be for
a different value of $n \geq 4$. We omit the trivial cases where $k=1$ or
$\delta=2$. We also consider only the cases where $k \leq n-k$ and
ignore the cases where the size of the optimal code is one. By
Theorem~\ref{thm:exactk_2k} we have that $\cA_2(4,2,2)=5$; by
Theorems~\ref{thm:upperk_2k} and~\ref{thm:lowerk_2k} we have
$\cA_2(5,2,2)=9$. Hence, the first table is for $n=6$.

\begin{center}

\hspace{4.00ex}
\begin{tabular}{|c||c@{\hspace{0.50ex}}%
                    c|}
\multicolumn{3}{c}{\sf \hspace*{-10ex}Bounds on
$\cA_2(6,\delta,k)$\hspace*{-10ex}}
\vspace{0.5cm}
\\[0.75ex]
\hline \raisebox{-1.25ex}{\large$\delta$} &
\multicolumn{2}{c|}{\large$k$\Strut{2.50ex}{0ex}}
\\[-1.00ex]
& 3 & 2
\\
\hline\hline $3$ & $^a9^a$\Strut{2.50ex}{0ex} &
\multicolumn{1}{|c|}{}
\\
\cline{3-3} $2$ & $^b77-81^c$ & $^a21^a$\Strut{2.25ex}{0ex}
\\
\hline
\end{tabular}

\vspace{0.7cm}

\hspace{4.00ex}
\begin{tabular}{|c||c@{\hspace{0.50ex}}%
                    c|}
\multicolumn{3}{c}{\sf \hspace*{-10ex}Bounds on
$\cA_2(7,\delta,k)$\hspace*{-10ex}}
\vspace{0.5cm}
\\[0.75ex]
\hline \raisebox{-1.25ex}{\large$\delta$} &
\multicolumn{2}{c|}{\large$k$\Strut{2.50ex}{0ex}}
\\[-1.00ex]
& 3 & 2
\\
\hline\hline $3$ & $^d17^d$\Strut{2.50ex}{0ex} &
\multicolumn{1}{|c|}{}
\\
\cline{3-3} $2$ & $^e329-381^f$ & $^g31^h$\Strut{2.25ex}{0ex}
\\
\hline
\end{tabular}

\vspace{0.7cm}

\begin{tabular}{|c||c@{\hspace{0.50ex}}%
                    c@{\hspace{0.50ex}}%
                    c|}
\multicolumn{4}{c}{\sf \hspace*{-10ex}Bounds on
$\cA_2(8,\delta,k)$\hspace*{-10ex}}
\vspace{0.5cm}
\\[0.75ex]
\hline \raisebox{-1.25ex}{\large$\delta$} &
\multicolumn{3}{c|}{\large$k$\Strut{2.50ex}{0ex}}
\\[-1.00ex]
& 4 & 3 & 2
\\
\hline\hline $4$ & $^a17^a$\Strut{2.50ex}{0ex} &
\multicolumn{2}{|c|}{}
\\
\cline{3-3} $3$ & $^i257-289^j$ & $^d34^d$\Strut{2.25ex}{0ex}
&\multicolumn{1}{|c|}{}
\\
\cline{4-4} $2$ & $^k4797-6477^f$ & $^e1312-1493^c$ & $^a85^a$\Strut{2.25ex}{0ex}
\\
\hline
\end{tabular}

\end{center}

\noindent
\begin{itemize}
\item $a$ - Theorem~\ref{thm:exactk_2k}.

\item $b$ - \cite{KoKu08}.

\item $c$ - \cite{EtVa11}.

\item $d$ - \cite{EJSSS}.

\item $e$ - \cite{BrRe12}.

\item $f$ - Theorem~\ref{thm:packing}.

\item $g$ - Theorem~\ref{thm:lowerk_2k}.

\item $h$ - Theorem~\ref{thm:upperk_2k}.

\item $i$ - The Multilevel Construction with $\CMRD$.

\item $j$ - Theorem~\ref{thm:Joh}.

\item $k$ - \cite{EtSi13}.
\end{itemize}

Some of the specific values of $\cA_q (n,\delta,k)$
can be of special interest. Some of these are discussed
in Sections~\ref{sec:Steiner} and~\ref{sec:spreads},
but there are some other as well. For example, the value
of $\cA_2 (2k,k-1,k)$, $k \geq 3$, is one such value.
By Theorem~\ref{thm:S_exactk_2k} we have $\cA_2 (2k-1,k-1,k-1) = 2^k+1$.
Hence, by Theorem~\ref{thm:Joh} we have
$\cA_2 (2k,k-1,k) \leq \frac{2^{2k}-1}{2^k-1} (2^k+1) =(2^k+1)^2$.
If there is a code which attains this bound then it should have
a very interesting symmetry even so a code which attains the
bound $\cA_2 (2k-1,k-1,k-1) = 2^k+1$ does not have a symmetry.
An interesting related question will be discussed in Section~\ref{sec:subspace}.

\begin{problem}
Is there some $k \geq 3$ for which $\cA_2 (2k,k-1,k) = (2^k+1)^2$ ?
Does $\cA_2(6,2,3)=81$ ?
\end{problem}

\begin{problem}
\label{pr:almost_Steiner}
Is there some $k \geq 4$ for which $\cA_2(2k,k-1,k) \geq (2k+1)^2 - 2^{k-1}$?
This bound will be of special interest in Section~\ref{sec:subspace}.
\end{problem}

\section{The Multilevel Construction}
\label{sec:multiL}

The multilevel construction is a method for which the outcome
is a code in $\cP_q(n)$ which contains lifted rank-metric codes. The
method can be used to construct Grassmannian codes, subspace
codes with the subspace distance, and subspace codes with the
injection distance. The construction is based on a new type of
rank-metric codes, namely Ferrers diagram rank-metric codes.
The description of the construction requires some
methods to represent subspaces of $\cP_q(n)$.

A $k$-dimensional subspace $X$ of $\F_q^n$ can be represented by a
$k \times n$ \emph{generator matrix} whose rows form a basis for $X$.
The basis of $X$ is in \emph{reduced row echelon
form}, denoted by $E(X)$, if the following conditions are satisfied.
\begin{itemize}
\item The leading coefficient of a row is always to the right of
the leading coefficient of the previous row.

\item All leading coefficients are {\it ones}.

\item Every leading coefficient is the only nonzero entry in its
column.
\end{itemize}

Each $k$-dimensional subspace $X$ of $\F_q^n$ has an
\emph{identifying vector} $v(X)$. $v(X)$ is a binary vector of length
$n$ and weight $k$, where the {\it ones} in $v(X)$ are in the
positions (columns) where $E(X)$ has the leading {\it ones} (of
the rows).

A \emph{Ferrers diagram} represents partitions of
positive integers as patterns of dots
with the $i$-th row having the same number of dots as the $i$-th
term in the partition~\cite{AnEr04,Sta86,vLWi92}. A Ferrers
diagram satisfies the following conditions.
\begin{itemize}
\item
The number of dots in a row is at most the number of dots in the
previous row.

\item All the dots are shifted to the right of the diagram.
\end{itemize}
The \emph{number of rows (columns)} of the Ferrers diagram $\cF$ is
the number of dots in the rightmost column (top row) of $\cF$. If
the number of rows in the Ferrers diagram is $m$ and the number of
columns is $\eta$ we say that it is an $m \times \eta$ Ferrers
diagram.
If we read the Ferrers diagram by columns we get another partition
of the same integer.

The \emph{echelon Ferrers form} of a vector $v$ of length $n$ and
weight $k$, $EF(v)$, is the $k\times n$ matrix in reduced row
echelon form with leading entries (of rows) in the columns indexed
by the nonzero entries of $v$ and $"\bullet"$  in all entries
which do not have terminals {\it zeroes} or {\it ones}. A
$"\bullet"$ will be called in the sequel a {\it dot}. The dots of this
matrix form the Ferrers diagram of $EF(v)$. If we substitute
elements of $\F_q$ in the dots of $EF(v)$ we obtain a
$k$-dimensional subspace $X$ of $\cG_q(n,k)$. $EF(v)$ will be called also
the echelon Ferrers form of $X$.

Let $v$ be a vector of length $n$ and weight $k$ and let $EF(v)$
be its echelon Ferrers form. Let $\cF$ be the Ferrers diagram of
$EF(v)$. $\cF$ is an $m \times \eta$ Ferrers diagram, where $m \leq k$
and $\eta \leq n-k$. A code $\cC$ is an $[\cF,\varrho,\delta]$ {\it
Ferrers diagram rank-metric code} if all codewords are $m \times
\eta$ matrices in which all entries not in $\cF$ are {\it zeroes},
it forms a rank-metric code with dimension $\varrho$, and minimum
rank distance $\delta$.

\vspace{1.8cm}

\noindent
{\Large{\bf The Multilelvel Construction}}

\vspace{0.3cm}

\noindent
{\bf First step:} choose a binary code
${\bf C}$ of length $n$ and
minimum Hamming distance $d$, $2\delta-1
\leq d \leq 2\delta$. This code will be called the
\emph{skeleton code}.

The next three steps are performed for each codeword
$c \in {\bf C}$.

\noindent
{\bf Second step:} construct the echelon Ferrers form
$EF(c)$.

\noindent
{\bf Third step:} construct an
$[\cF,\varrho,\delta]$ Ferrers diagram
rank-metric code $\cC_{\cF}$ for the Ferrers
diagram $\cF$ of $\mbox{EF}(c)$.

\noindent
{\bf Fourth step:} lift $\cC_{\cF}$
to an $(n,\delta,k)_q$ code
$\C_c$, for which the echelon Ferrers form of
$X \in \C_c$ is $\mbox{EF}(c)$.

\noindent
{\bf Finally:}
$$
\C = \bigcup_{c \in {\bf C}} \C_c ~.
$$
\begin{theorem}
The size of the code $\C$ is $\sum_{c \in {\bf C}} | \C_c |$ and
it has minimum subspace distance $d$.
\end{theorem}

The code $\C$ obtained by the Multilevel Construction
can be a Grassmannian code, a subspace code with the
subspace distance, or a subspace code with the
injection metric. If we want $\C$ to be a Grassmannian code then
${\bf C}$ must be a constant weight code in the Hamming scheme
(or equivalently a code in the Johnson scheme). The other two
options will be discussed in the next two sections.
A comprehensive description and discussion
on the Multilevel Construction can be found in~\cite{EtSi09}.
Some improvements which enable to add more codewords
for the final code can be found in~\cite{SiTr13,TrRo10}.

\begin{problem}
Find a method which combine the Multilevel Construction
with more concepts to obtain larger codes.
\end{problem}

\begin{problem}
What is the best way to choose a skeleton code for the Multilevel Construction.
\end{problem}

\section{Subspace Codes with the Subspace Distance}
\label{sec:subspace}

An $(n,d)_q^S$ code is
a subspace code in $\cP_q(n)$ with minimum subspace distance $d$.
Let $\cA_q^S (n,d)$ be the maximum number of codewords in
an $(n,d)_q^S$ code. Recall that the Grassmannian distance is the
$q$-analog of the Johnson distance and that the subspace distance
is the $q$-analog of the Hamming distance. But, while the related Johnson
graph, Grassmannian graph, and Hamming graph are distance regular,
the related graph of $\cP_q(n)$, $G(\cP_q^S(n))$, is not
distance regular. In fact, it is not regular since the sizes of the balls
with radius one around a vertex depends on the
dimension of the related subspace. This makes the task, of handling coding questions
in $\cP_q(n)$ in general and obtaining bounds on $\cA_q^S(n,d)$ in particular,
more difficult. Even so, classic lower bounds on $\cA_q^S (n,d)$
such as the Gilbert-Varshamov bound
were given in~\cite{EtVa11}.

The Multilevel Construction can be applied to
obtain subspace codes with the subspace
distance, where the skeleton code is taken to be a binary code
in the Hamming scheme. Subspace codes can be also cyclic codes.
Hence, a method to construct cyclic subspace codes is
one of the important tasks in this area. An example
of a cyclic $(6,3)_2^S$ code of size 85 was given in~\cite{EtVa11}.
Construction of cyclic codes in $\cP_q(n)$ is a method which
should be further explored.
A third method to construct subspace codes was presented in~\cite{EtSi09}.
This is a puncturing method which is used to obtain codes with subspace distance
$d-1$ from codes with subspace distance $d$. This method works as follows,

Let $X$ be an $\ell$-subspace of $\F_q^n$ such that the unity
vector with an {\it one} in the $i$-th coordinate is not an element
in $X$. The \emph{puncturing} of the $i$-th coordinate of $X$, $\Delta_i(X)$,
is defined as the $\ell$-dimensional subspace of $\F_q^{n-1}$
obtained from $X$ by deleting coordinate $i$ from each vector in $X$.
Let $\C$ be a code in $\cP_q(n)$ and let $Q$ be an $(n-1)$-dimensional
subspace of $\F_q^n$. Let $E(Q)$ be the $(n-1) \times n$ generator
matrix of $Q$ (in reduced row echelon form) and let $\tau$ be the
position of the unique {\it zero} in $v(Q)$. Let $v \in \F_q^n$ be an
element such that $v \notin Q$. We define the {\it punctured} code

$$\C'_{Q,v} = \C_Q \cup \C_{Q,v}~,$$
where
$$\C_Q=\left\{ \Delta_\tau (X) ~:~ X \in \C, ~ X \subseteq Q\right\})$$
and
$$\C_{Q,v}= \left\{ \Delta_\tau (X \cap Q)~:\: X \in \C,~v\in X \right\}~.$$

This method is very useful and
will usually produce codes which are larger than the
codes generated by the Multilevel Construction when
$d$ is even and the resulting code has an odd subspace distance. It is
also more effective in terms of good bounds when we start with a code
of $\F_q^n$ for even $n$, in which a large $(n,\delta,\frac{n}{2})_q$ code
is contained, and the resulting code contains subspaces
of $\F_q^{n-1}$. Sometimes it would be better to start only with the
related Grassmannian code $\C$ and to puncture $\C$. To the punctured
code $\C'_{Q,v}$ codewords with other dimensions should be added.
The right way, when and how to apply the puncturing method is
a topic for future research.

The only reasonable known methods to obtain upper bounds on $\cA_q^S (n,d)$ are
various types of linear programming
methods. The linear programming method suggested in~\cite{EtVa11}
is different from the classic linear programming of Delsarte~\cite{Del73}.
The classic linear programming does not yield any improvements,
when applied for bounds on codes in $\cG_q(n,k)$
or in $\cP_q(n)$. A somehow
stronger method is the semidefinite programming~\cite{Sch79,Sch05}
as suggested in~\cite{BPV13} for the projective space.
This method and the one in~\cite{EtVa11}
yield all the best known upper bounds on $\cA_q^S (n,d)$,
when in most cases the semidefinite programming yields better bounds.
The following list introduces a small list of the related research problems.

\begin{problem}
Find a counterpart Sphere packing bound on $\cA_q^S(n,d)$.
\end{problem}

\begin{problem}
Find a $q$-analog for the well-known Plotkin bound on $\cA_q^S (n,d)$
and analyze the related codes which attain the bound.
\end{problem}

\begin{problem}
Can additional constraints be added for the semidefinite programming
to improve the upper bounds on $\cA_q^S (n,d)$?
\end{problem}

\begin{problem}
Find a new method to construct large subspace codes with the subspace distance.
\end{problem}

\begin{problem}
Find a method to construct large cyclic codes in $\cP_q(n)$.
\end{problem}

\begin{problem}
\label{pr:tab_sub}
Compile a table with lower and upper bounds on $\cA_q^S (n,d)$ for small $q$, $n$, and $d$.
\end{problem}

Let's consider Research problem~\ref{pr:tab_sub} for $q=2$ and $4 \leq n \leq 7$.
$\cA_2^S (n,1) = | \cP_q(n)|$
and $\cA_2^S(n,2)=\Sigma_{\text{even}~k} \sbinomtwo{n}{k}$, where
$\cA_2^S (n,2) =\frac{|\cP_q(n)|}{2}$ for odd $n$.
$\cA_2^S (4,3) =\cA_2^S(4,4)=5$, where the lower bound is
derived from a spread (see Section~\ref{sec:spreads})
and the upper bound is obtain using trivial analysis.
It was proved that $\cA_2^S(5,3)=18$ in~\cite{EtVa11},
by Theorem~\ref{thm:b_2n+1_2n} we have $\cA_2^S(5,4)=9$,
and it is easy to verify that $\cA_2^S (5,5)=2$.
$\cA_2^S(6,3) \leq 123$ as was proved in~\cite{EtVa11}, where a lower
bound of $\cA_2^S(6,3) \geq 85$ was also given. The lower bound can be probably
improved considerably by considering the 77 codewords of the $(6,2,3)_2$
code reported in~\cite{KoKu08} and adding to it subspaces with dimensions
2 and 4. The actual code size should be found by using
a computer search. We leave it as we leave
other basic computations which follow to the
interested reader. $\cA_2^S(6,4) = \cA_2(6,2,3)$
and $\cA_2^S(6,5)=\cA_2^S(6,6)=\cA_2(6,3,3)=9$. $\cA_2^S(7,3 ) \leq 776$
by using semidefinite programming as was proved in~\cite{BPV13}. The lower
bound $\cA_2^S(7,3) \geq 584$ is obtained by puncturing the $(8,2,4)_2$
code of size 4797~\cite{EtSi13} and adding to it
the null space and $\F_2^7$. It is not difficult to verify
based on previous results which were stated that
$330 \leq \cA_2^S(7,4) = \cA_2(7,2,3)+1 \leq 382$.
For the next two cases, $\cA_2^S(7,5)$ and $\cA_2^S(7,6)$ we will
provide more detailed and general analysis.

\begin{lemma}
\label{lem:half_dim}
If $\C$ is a $(2n+1,n,n+1)_2$ code of size $2^{n+1}+1$
then each one-dimensional subspace of $\F_2^{2n+1}$ is contained
in one of its codewords.
\end{lemma}
\begin{proof}
Assume the contrary, i.e. that no subspace of $\C$ contains
a given vector $v \in \F_2^{2n+1}$. Let $V$ be the one-dimensional
subspace of $\F_2^n$ which contains $v$. Hence, $\C^\perp$
is an $(2n+1,n,n)_2$ code of size $2^{n+1}+1$ in which no
codeword is contained in the $(2n)$-dimensional subspace $V^\perp$.
Therefore, for each codeword $X \in \C^\perp$ we have $dim (X \cap V^\perp)=n-1$
and $| X \cap (\F_2^{2n+1} \setminus V^\perp ) | =2^{n-1}$.
Clearly, for each two codewords $X, Y \in \C^\perp$ we have
$\dim (X \cap Y) =0$. Therefore,
$$2^{2n} = | \F_2^{2n+1} \setminus V^\perp | \geq | \bigcup_{X \in \C^\perp} ( X \cap V^\perp )| = \sum_{X \in \C^\perp} 2^{n-1} = 2^{2n} +2^{n-1}~,$$
a contradiction.
\end{proof}

By Lemma~\ref{lem:half_dim}, since $\cA_2(2n+1,n,n+1)=\cA_2(2n+1,n,n)=2^{n+1}+1$
(see Theorem~\ref{thm:S_exactk_2k}),
and simple distance analysis (we also have to use
the orthogonal complement code) we have

\begin{theorem}
\label{thm:b_2n+1_2n}
$\cA_2^S(2n+1,2n) = 2^{n+1} +1$.
\end{theorem}

\begin{theorem}
$2^{n+2}+1 \leq \cA_2^S (2n+1,2n-1) \leq 2^{n+2} +2$.
\end{theorem}
\begin{proof}
The upper bound is implied by the same arguments as
the ones for
the proof of Theorem~\ref{thm:b_2n+1_2n}.

The lower bound is derived by puncturing a $(2n+2,n,n+1)_2$
$\CMRD$ code $\C$ and adding one codeword to the punctured code.

The size of $\C$ is $2^{2n+2}$ and all the
$2^{n+1}-1$ nonzero vectors of $\F_2^{2n+2}$
which start with $n+1$ \emph{zeroes} are not contained
in any codeword of $\C$. Each other
nonzero vector of $\F_2^{2n+2}$
is contained in the same number of codeword (see~\cite{EtSi13}).
Each codeword contains $2^{n+1}-1$ nonzero vectors of $\F_2^{2n+1}$.
Hence, each of these $2^{2n+2}-2^{n+1}$ vectors is contained
in $\frac{2^{2n+2}(2^{n+1}-1)}{2^{2n+2}-2^{n+1}}=2^{n+1}$ codewords of $\C$.

Now, we compute the number of $(2n+1)$-dimensional subspaces
of $\F_2^{2n+2}$ which contain a given $(n+1)$-dimensional
subspace $X$ of $\F_2^{n+1}$. The number of such
$(2n+1)$-dimensional subspaces is
$$\frac{(2^{2n+2}-2^{n+1})(2^{2n+2}-2^{n+2}) ~ \cdots ~ (2^{2n+2}-2^{2n})}
{(2^{2n+1}-2^{n+1})(2^{2n+1}-2^{n+2}) ~ \cdots ~ (2^{2n+1}-2^{2n})}
=\frac{(2^{2n+2}-2^{n+1})2^{n-1}}{2^{2n+1}-2^{2n}} =2^{n+1}-1~.$$

Hence, the number of $(2n+1)$-dimensional subspaces which do not
contain all the
$2^{n+1}-1$ nonzero vectors of $\F_2^{2n+2}$
which start with $n+1$ \emph{zeroes} is
$2^{2n+2}-1 - (2^{n+1} -1) = 2^{2n+2} -2^{n+1}$.

Therefore, there exists at least one $(2n+1)$-dimensional
subspace $Q$ of $\F_2^{2n+2}$ which
contains $\frac{2^{2n+2}(2^{n+1}-1)}{2^{2n+2} -2^{n+1}} =2^{n+1}$
codewords (~$(n+1)$-dimensional subspaces) of $\C$.

Thus, the punctured code $\C'_{Q,v}$, where $v$ is any nonzero vector
not in $Q$ has size $2^{n+2}$. To $\C'_{Q,v}$ we can add either any
$n$-dimensional subspace which contains all the $2^n-1$
vectors which start with $n$ \emph{zeroes} or any $(n+1)$-dimensional
subspace which contains  all these vectors, without damaging the minimum
subspace distance.

Thus, $2^{n+2}+1 \leq \cA_2^S (2n+1,2n-1) \leq 2^{n+2} +2$.
\end{proof}

Note, that for any $k$ for which
$\cA_2(2n+2,n,n+1) \geq (2n+3)^2 - 2^n$
(see Research problem~\ref{pr:almost_Steiner})
we will have $\cA_2^S (2n+1,2n-1) \leq 2^{n+2} +2$.

\begin{problem}
Find good bounds on the size of the $(n,d)_q^S$ code,
$d=2 \delta +1$, which contains
subspaces only with dimension $k$ and $k-1$,
where $\delta +2 \leq k \leq \lceil \frac{n}{2} \rceil$.
\end{problem}

\begin{problem}
Find good bounds on the size of the
$(2k,d)_q^S$ code, $d$ odd, which contains subspaces
only with dimensions $k-1$, $k$, and $k+1$.
\end{problem}

\begin{problem}
Find a $(7,5)_2^S$ code with 34 codewords or prove that
$\cA_2^S(7,5) < 34$.
\end{problem}

\begin{problem}
Determine whether $\cA_2^S (2n+1,2n-1) = 2^{n+2} +1$
or $\cA_2^S (2n+1,2n-1) = 2^{n+2} +2$.
\end{problem}

\begin{problem}
What is the size of the largest equidistant $(n,d)_q^S$ code?
\end{problem}

\begin{problem}
Is the size of the
largest equidistant $(n,d)_q^S$ code is greater by one from the
size of the largest equidistant $(n,\lfloor \frac{n}{2},d)_q$ code?
\end{problem}

Another question of interest is the existence question of perfect
subspace codes with the subspace distance. It was proved in~\cite{EtVa11}
that such codes do not exist. The
Johnson space and the Grassmann space admit
\emph{diameter-perfect} codes~\cite{AAK01}. All such
diameter-perfect codes are optimal for their parameters.
Unfortunately, the definition of diameter-perfect codes does not
extend to the projective space $\cP_q(n)$, since the size of a
sphere~in $\cP_q(n)$ depends on its center.

\begin{problem}
Can one define another type
of perfect codes in the projective space, so that certain optimal
codes become ``perfect'' under this definition?
\end{problem}

\section{The Injection Metric}
\label{sec:injection}

An $(n,d)_q^I$ code is a subspace code in $\cP_q(n)$
with minimum injection distance $d$.
The injection distance is the one which is more useful,
than the subspace distance, from a practical
point of view~\cite{SiKs09}.
Let $\cA_q^I (n,d)$ denote the maximum number of codewords
in an $(n,d)_q^I$ code. The injection distance is the $q$-analog
of the asymmetric distance~\cite{KhKs09}.
Also, the related graph $G(\cP_q^I(n))$ is not
distance regular which makes the analysis of some
bounds more difficult as in the case of the
subspace distance and the related
graph $G(\cP_q^S(n))$. Similarly to the Grassmannian codes
and the subspace codes with the subspace distance,
we can generate a large code for the injection
metric by using the Multilevel Construction, where our skeleton code is
a binary code with the asymmetric distance~\cite{KhKs09}. Instead of
Hamming distance $d$ for the skeleton code in the Multilevel Construction,
we use asymmetric distance $\delta$. Large asymmetric codes can be found
for example in~\cite{Etz91,EtOs98,Shi02}. Puncturing is also useful to obtain good codes
with the injection distance. But, puncturing constant dimension codes
yields the same subspace codes with the injection distance as the ones
with the subspace distance. Therefore, puncturing might be useful,
to obtain good codes with the injection distance,
when we start from a subspace code with the injection distance
which is larger then the related code with the subspace distance.
Another option is to puncture a Grassmannian code and to the
punctured code, codewords with permitted dimensions (by the
required injection distance) are added. We note that the permitted
dimensions are more dense than in the case of a similar subspace
code with the subspace distance. This will be further
demonstrated and explained at
the end of this section. Also, as in the case of the subspace distance,
classic lower bounds such as the Gilbert-Varshamov bound were developed
in~\cite{GaYa10,KhKs09}. As for upper bounds, a linear programming
was developed in~\cite{AhAy}, and in~\cite{BPV13} there is a modification
of the linear programming given in~\cite{EtVa11}. A semidefinite
programming bound was also given in~\cite{BPV13} and it was shown that the
bounds obtained are very similar to those obtained by linear programming.

\begin{problem}
Find a counterpart Sphere packing bound on $\cA_q^I(n,d)$.
\end{problem}

\begin{problem}
\label{pr:tab_inj}
Compile a table with lower and upper bounds on $\cA_q^I (n,d)$ for small $q$, $n$, and $d$.
\end{problem}

\begin{problem}
Find a new method to construct large subspace codes with the injection distance,
which differ from the ones used for subspace codes with the subspace distance.
\end{problem}

\begin{problem}
Make a comprehensive comparison between the injection metric and the subspace metric,
beyond consequences related to the asymmetric distance and the
Hamming distance, respectively.
\end{problem}

For very small values of $n$ related to a given
injection distance $d$, the tables of Research problems~\ref{pr:tab_sub}
and~\ref{pr:tab_inj} are very similar. For example, while
$\cA_2^S (4,4)=5$, the related bound for the injection distance is
$\cA_2^I(4,2)=7$ since we can add the null space and $\F_2^4$
to five subspaces of dimension two (since subspace codes with injection
distance $d$ can be denser than their related subspace
codes with subspace distance $2d-1$). But, in some cases,
we can find much larger
codes with nice structure for the injection distance.
As an example we consider $n=6$ and injection distance 2.
Let $\alpha$ be a root of $x^6+x+1$,
and use this primitive polynomial to generate $\F_{64}$. Consider
the code $\C$ in $\cP_2(6)$ which consists of all the cyclic
shifts of $$\{ {\bf 0} , \alpha^0 , \alpha^{21} , \alpha^{42} \}~,$$
$$\{ {\bf 0} , \alpha^0 , \alpha^1 , \alpha^4 , \alpha^6 , \alpha^{16} ,
\alpha^{24}, \alpha^{33} \}~,$$ and
$$\{ {\bf 0} , \alpha^0 , \alpha^1 , \alpha^6 , \alpha^8 , \alpha^{18} ,
\alpha^{21}, \alpha^{22}, \alpha^{27}, \alpha^{29} , \alpha^{39},
\alpha^{42}, \alpha^{43}, \alpha^{48}, \alpha^{50} , \alpha^{60} \}~.$$
To the 105 subspaces obtained in this process
we add the null space and $\F_2^6$ to obtain
a cyclic $(6,2)_2^I$ code of size 107 which implies that $\cA_2^I(6,2) \geq 107$,
which is much better than what we can obtain for a related $(6,3)_2^S$ code.
The upper bounds in both cases are very similar,
since by linear programming we have $\cA_2^I(6,2) \leq 125$,
while $\cA_2^S(6,2) \leq 123$.

\begin{problem}
Find good bounds on the size of a
$(2k,d)_q^I$ code, $d$ odd, which contains subspaces
only with dimensions $k-1$, $k$, and $k+1$. Compare
the bounds with the related bounds on the
large size of a $(2k,2d-1)_q^S$ code.
\end{problem}

\begin{problem}
Find large cyclic subspace codes with the injection distance.
\end{problem}

\section{$q$-Analog of Steiner Systems}
\label{sec:Steiner}

A \emph{Steiner system} $S(t,k,n)$ is a collection $S$ of $k$-subsets
from an $n$-set $\cN$ such that each $t$-subset of $\cN$ is contained
in exactly one element of $S$. Steiner systems were subject to an
extensive research in combinatorial designs~\cite{CoDi07}.
A Steiner system is
also equivalent to an optimal constant weight codes in
the Hamming scheme.

Cameron~\cite{Cam74,Cam74a} and Delsarte~\cite{Del76} have
extended the notions of block design and Steiner systems
to vector spaces.
A \emph{Steiner structure} (\emph{$q$-Steiner system})
$\dS_q(t,k,n)$ is a collection $\dS$ of
elements from $\cG_q(n,k)$ (called \emph{blocks}) such that each element from
$\cG_q(n,t)$ is contained in exactly one block of $\dS$.
A Steiner structure $\dS_q(t,k,n)$ is a constant dimension
code which attains the bound of $\cA_q(n,k-t+1,k)$.
Similarly, to Steiner systems, simple necessary divisibility conditions
for the existence of a given Steiner structure are developed~\cite{ScEt}.

\begin{theorem}
\label{thm:derived}
If a Steiner structure $\dS_q (t,k,n)$ exists then
for each $i$, $1 \leq i \leq t-1$, a Steiner structure
$\dS_q(t-i,k-i,n-i)$ exists.
\end{theorem}

\begin{cor}
If a Steiner structure $\dS_q(t,k,n)$ exists then for all $0 \leq i \leq t-1$.
$$
\frac{\sbinomq{n-i}{t-i}}{\sbinomq{k-i}{t-i}}
$$
must be integers.
\end{cor}

Steiner structures and Steiner systems are highly related.
In~\cite{EtVa11a,ScEt} there are some constructions of Steiner systems
derived from Steiner structures. Further research on Steiner structures
seems to be fascinating, but also extremely difficult. We list some interesting,
but probably very difficult research problems.

\begin{problem}
Let $\dS$ be a Steiner structure $\dS_q(t,k,n)$.
Find new sets of parameters $t' < k' < n'$
for which there exists a Steiner system $S(t',k',n')$
derived from $\dS$.
\end{problem}

\begin{problem}
Are there more Steiner structures embedded in a Steiner structure
$\dS_q(t,k,n)$, except from the ones implied by Theorem~\ref{thm:derived}.
\end{problem}

Until recently, the only known Steiner structures $\dS_q(t,k,n)$ were either trivial or
for $t=1$, where such structures exist if and only if $k$ divides $n$.
These are called spreads and they will be discussed in the next section.
Thomas~\cite{Tho96} showed that certain kind of Steiner
structures $\dS_2(2,3,7)$ cannot exist. Metsch~\cite{Met99} conjectured
that nontrivial Steiner structures with $t \geq 2$ do not exist.
Steiner structures appear also in connection of diameter perfect
codes in the Grassmann scheme. It was proved in~\cite{AAK01}
that the only diameter perfect codes in the Grassamnn scheme
are the $q$-Steiner systems.
The following theorem given in~\cite{EtVa11a} has given more
indication that finding Steiner structures with $t \geq 2$ would
be a very difficult task.

\begin{theorem}
\label{thm:StoS}
If there exists a Steiner structure $\dS_2 (2,k,n)$ then there
exists a Steiner system $S(3,2^k,2^n)$.
\end{theorem}
As a consequence of Theorem~\ref{thm:StoS},
we have that if there exists a Steiner structure
$\dS_2(2,3,7)$ then there exists a Steiner system $S(3,8,128)$.
The existence of a Steiner system $S(3,8,128)$ is an open problem,
which might strengthen the conjecture that a Steiner structure
$\dS_2(2,3,7)$ does not exist.

Recently, the first Steiner structure $\dS_q(t,k,n)$
with $t \geq 2$ was found.
This is a Steiner structure $\dS_2(2,3,13)$ which have a large
automorphism group~\cite{BEOVW}. We will describe this group in terms of two mappings
and an equivalence relation defined on  $k$-dimensional subspaces.
Let $\alpha$ be a primitive element in $\F_{q^n}$ and define
the following two mappings

The \emph{Frobenius mapping} $\Upsilon_\ell$,
$0 \leq \ell \leq n-1$, $\Upsilon_\ell : \F_{q^n} \setminus \{ 0 \}
\longrightarrow \F_{q^n} \setminus \{ 0 \}$
is defined by $\Upsilon_\ell (x) \deff x^{q^\ell}$ for each
$x \in \F_{q^n} \setminus \{ 0 \}$.

The \emph{cyclic shift mapping} $\Phi_j$, $0 \leq j \leq q^n-2$,
$\Phi_j : \F_{q^n} \setminus \{ 0 \}
\longrightarrow \F_{q^n} \setminus \{ 0 \}$ is defined by
$\Phi_j (\alpha^i) \deff \alpha^{i+j}$, for each $0 \leq i \leq q^n-2$.

The two types of mappings $\Upsilon_\ell$ and $\Phi_j$ can be applied
on a subset or a subspace, by applying the mapping on each
element of the subset or the subspace, respectively. Formally,
given two integers $0 \leq \ell \leq n-1$ and $0 \leq j \leq q^n-2$,
$$
\Upsilon_\ell \{ x_1 , x_2 , \ldots , x_r \} \deff
\{ \Upsilon_\ell (x_1), \Upsilon_\ell (x_2),\ldots,\Upsilon_\ell(x_r) \}~,
$$
$$
\Phi_j \{ x_1 , x_2 , \ldots , x_r \} \deff \{ \Phi_j (x_1), \Phi_j
(x_2),\ldots,\Phi_j (x_r) \}.
$$

We define the following equivalence relation $\tilde{E}$ on the
subspaces of $\cG_q(n,k)$.
$$
(X,Y) \in \tilde{E} ~~~ \text{if} ~~~ \text{there~exist~two~integers},~\ell_1,~j_1,~ \text{such~that}~ Y= \Phi_{j_1} (\Upsilon_{\ell_1} (X))~.
$$

In~\cite{BEOVW} Steiner structures $\dS_q(t,k,n)$, in
which a $k$-dimensional
subspace $X$ is contained in the system if all its equivalence class,
under the equivalence relation $\tilde{E}$,
is contained in the structure, are considered.
Such structures can be described with a relatively small
number of representatives. The Steiner structure
$\dS_2(2,3,13)$ constructed in~\cite{BEOVW} was constructed
in this way and only 15 representatives describe the whole system.
The knowledge on $q$-Steiner systems is very small and there are many more
research problems for future research. As we noted before, these problems are
probably extremely difficult. The next few problems to consider are as follows.

\begin{problem}
Prove or disprove that a $q$-Steiner system $\dS_2(2,3,7)$ exists.
\end{problem}

\begin{problem}
Construct more $q$-Steiner systems $\dS_2(2,3,n)$, where $n \equiv 1~(\text{mod}~6)$
is a prime.
\end{problem}

\begin{problem}
By using the equivalence relation $\tilde{E}$,
develop the theory and find an example of a $q$-Steiner system $\dS_2(2,4,n)$
for some $n$; and $\dS_q(2,3,n)$ for some $q>2$ and some $n$.
\end{problem}

\begin{problem}
Find parameters for which the necessary conditions for
the existence conditions for a $q$-Steiner systems
$\dS_q(t,k,n)$ are satisfied, but the systems do not exist.
\end{problem}

\begin{problem}
Find new necessary conditions for the existence of
$q$-Steiner systems.
\end{problem}

We conclude this section with a method that can lead for the construction
of a $q$-Steiner system $\dS_2(2,3,7)$ or to a proof for its nonexistence.
This method will be called the \emph{projections method} and it is described
for a general $q$-Steiner system $\dS_q(t,k,n)$.

We start by assuming that $\dS$ is a $q$-Steiner system
$\dS_q(t,k,n)$. Each $r$-dimensional subspace~$Z$
of $\F_q^n$ will be represented
by an $n \times \frac{q^r-1}{r-1}$ matrix whose columns
represent the nonzero vectors contained in~$Z$, where the leading
nonzero element in a vector is an \emph{one}. For a given~$\rho$,
$1 \leq \rho \leq n$, we construct a system $\dS_\rho$, which consists
of the projections of the first~$\rho$ rows of each subspace of $\dS$.
Each $\rho \times \frac{q^r-1}{r-1}$ matrix formed in this way represents
an $\ell$-dimensional subspace of $\F_q^\rho$
for some $\ell$, $0 \leq \ell \leq \min \{ \rho,k \}$.
Similarly, the projection of the first $\rho$ rows in
a $t$-dimensional subspace is an $\ell$-dimensional subspace
of $\F_q^\rho$ for some $\ell$, $0 \leq \ell \leq \min \{ \rho,t \}$.
For each $i$, $0 \leq i \leq \min \{ \rho ,k \}$, we have
$\sbinomq{\rho}{i}$ variables. A variable $a_Y$ for each
$i$-dimensional subspace $Y$ of $\F_q^\rho$. The value
of a variable is the number of times the related $i$-dimensional
subspace appears in the system $\dS_\rho$.
For each $i$, $0 \leq i \leq \min \{ \rho,t \}$,
we generate $\sbinomq{\rho}{i}$ equations. An
equation for each $i$-dimensional subspace of $\F_q^\rho$.
Let $X$ be such an $i$-dimensional subspace of $\F_q^\rho$.
Let $\delta_X$ the number of distinct ways to complete $X$
into a $t$-dimensional subspace of~$\F_q^n$.
Let $\Gamma_{X,Y}$ be the number of times that $X$ is contained in
an $\ell$-dimensional subspace $Y$ of~$\F_q^\rho$ such
that $i \leq \ell$ (taking into account
that $X$ is completed into a $t$-dimensional subspace and~$Y$
is completed into a $k$-dimensional subspace). For each such
$i$-dimensional subspace $X$ we generate the equation

$$
\delta_X = \sum_{\substack{{Y \in \F_q^\rho} \\ {\dim X \leq \dim Y}}} \Gamma_{X,Y} a_Y  ~.
$$
If the system of equations does not have a nonnegative integer solution then
a $q$-Steiner system $\dS_q(t,k,n)$ does not exist.

\begin{example}
Let $n=7$, $k=3$, $t=2$, and $\rho =2$. Let $a_0$ be the variable for the
null space $X_0$. Let $a_1$, $a_2$, $a_3$ be the
variable for the one-dimensional
subspaces of $\F_2^2$, $X_1 = \left\{ \begin{array}{c} 0\\ 1
\end{array} \right\}$, $X_2 = \left\{ \begin{array}{c} 1\\ 0
\end{array} \right\}$, $X_3 = \left\{ \begin{array}{c} 1\\ 1
\end{array} \right\}$, respectively. Let $a_4$ be the
variable for the two-dimensional
subspace $X_4 = \left\{ \begin{array}{ccc} 0 & 1 & 1\\
1 & 0 & 1
\end{array} \right\}$ of $\F_2^2$. One can easily verify that $\delta_{X_0} =
\sbinomtwo{5}{2} =155$, $\delta_{X_1} = \delta_{X_2} =\delta_{X_3}
= \binom{32}{2} = 496$, and $\delta_{X_4} = 32^2 = 1024$. Note,
that $\sbinomtwo{5}{2} + 3 \binom{32}{2} + 32^2 = \sbinomtwo{7}{2}
= 2667$. Now, let $Y_i = X_i$ for $0 \leq i \leq 4$.
We now have $\Gamma_{X_0,Y_0} =7$, $\Gamma_{X_0,Y_i}=1$,
$\Gamma_{X_i,Y_i} = 6$, $\Gamma_{X_i,Y_4}=1$, for $i=1,2,3$, and
$\Gamma_{X_4,Y_4} =4$. For any other $i$ and $j$ we have
$\Gamma_{X_i,Y_j}=0$.

Therefore, we have the following 5 equations:
\begin{enumerate}
\item $155= 7 a_0 + a_1 + a_2 + a_3$.

\item $496 = 6a_1 + a_4$.

\item $496 = 6a_2 + a_4$.

\item $496 = 6a_3 + a_4$.

\item $1024 = 4a_4$.
\end{enumerate}
This system of equations has exactly one solution, $a_0=5$,
$a_1=a_2=a_3=40$, and $a_4=256$.
\end{example}

If $n=7$, $k=3$, $t=2$, and $\rho =4$ then there are
$\sbinomtwo{4}{0} + \sbinomtwo{4}{1} + \sbinomtwo{4}{2} + \sbinomtwo{4}{3} =66$
variables and $\sbinomtwo{4}{0} + \sbinomtwo{4}{1} + \sbinomtwo{4}{2} =51$ equations.
From the 16 variables which represents the
null space and the 15 one-dimensional
subspaces exactly one variable should be equal \emph{one}. If we set one of these
16 variables to \emph{one} then there is a unique solution to the set of equations.
For example, if the variable which represents the null subspace
is equal to \emph{one} then each variable out of the 35 variables which represent
the 35 two-dimensional subspaces is equal to 4 and
each variable out of the 15 variables which represent
the 15 three-dimensional subspaces is equal to 16.

If $n=7$, $k=3$, $t=2$, and $\rho =5$ then there are
$\sbinomtwo{5}{0} + \sbinomtwo{5}{1} + \sbinomtwo{5}{2} + \sbinomtwo{5}{3} =342$
variables and $\sbinomtwo{5}{0} + \sbinomtwo{5}{1} + \sbinomtwo{5}{2} =187$ equations.
A computer program found a large number of solutions to these equations.

If $n=7$, $k=3$, $t=2$, and $\rho =6$ then there are
$\sbinomtwo{6}{0} + \sbinomtwo{6}{1} + \sbinomtwo{6}{2} + \sbinomtwo{6}{3} =2110$
variables and $\sbinomtwo{6}{0} + \sbinomtwo{6}{1} + \sbinomtwo{6}{2} =715$ equations.
We were not able to handle the system of equations, but this system
might be the key to settle the existence question of a $q$-Steiner system
$\dS_2(2,3,7)$.

\begin{problem}
Finish the projections method to settle the existence problem
of a $q$-Steiner system $\dS_2(2,3,7)$.
\end{problem}

\begin{problem}
Find more necessary conditions for the existence of a $q$-Steiner system
$\dS_q(t,k,n)$ by using the projections method.
\end{problem}

\begin{problem}
Apply the projections method on various parameters for
$q$-Steiner systems $\dS_q(t,k,n)$ either to exclude the existence of
some systems or to find more insight on the existence question.
\end{problem}

\section{Spreads and Partial Spreads}
\label{sec:spreads}

Two subspaces $X$, $Y$ of $\cP_q(n)$ are called \emph{disjoint} if
their intersection is the null space, i.e. $X \cap Y = \{ {\bf 0} \}$.
A \emph{spread} $\dS$ in $\cG_q(n,k)$ is a set of pairwise disjoint
subspaces of $\cG_q(n,k)$ in which each one-dimensional subspace
of $\cP_q(n)$ is a subspace of exactly one element of $\dS$.
A~\emph{partial spread} $\dS$ in $\cG_q(n,k)$ is a set of pairwise disjoint
subspaces of $\cG_q(n,k)$ in which each one-dimensional subspace
of $\cP_q(n)$ is a subspacce of at most one element of $\dS$.
A spread in the Grassmannian $\cG_q(n,k)$ is clearly
a constant dimension code with $\frac{q^n-1}{q^k-1}$
codewords and minimum Grassmannian distance $k$.
It attains the bound on $\cA_q(n,k,k)$ and it is
also a $q$-Steiner system $\dS_q(1,k,n)$. Any $(n,k,k)_q$
code is a \emph{partial spread}. A spread is also a well-known
and important concept in projective geometry. The projective
geometry PG($n,q$) consists of $\frac{q^{n+1}-1}{q-1}$
points and $\frac{(q^{n+1}-1)(q^n-1)}{(q^2-1)(q-1)}$ lines.
The points are represented by a set of
nonzero elements from $\F_q^{n+1}$, of maximum size,
in which each two elements are linearly independent.
Each element $x$ of these $\frac{q^{n+1}-1}{q-1}$
elements represents $q-1$
elements of $\F_q^{n+1}$ which are the multiples of $x$ by the nonzero
elements of $\F_q$. A line in PG($n,q$) consists of $q+1$ points.
Given two distinct points $x$ and $y$, there
is exactly one line which contains these two points.
This line contains $x$ and $y$ and the $q-1$ points
of the form $\gamma x +y$, where $\gamma \in \F_q \setminus \{ 0 \}$.
A point is a 0-subspace in PG($n,q$), a line is a 1-subspace in PG($n,q$),
and a $k$-subspace is constructed by taking a
$(k-1)$-subspace $Y$ and a point $x$ not on $Y$ and all points that
are constructed by a linear combination of $x$ with any set of points
from $Y$. We note that there is a difference of one in the definition between
the dimensions of subspaces in projective geometry and subspaces
in the projective space (or the Grassmannian). In the
sequel, we will continue to
use the notation used for the Grassmannian. A $k$-spread
in PG($n,q$) is a set of pairwise disjoint $k$-subspaces
of PG($n,q$) for which any point of PG($n,q$) is contained
in exactly one $k$-subspace. Such a $k$-spread is a spread
in $\cG_q(n+1,k+1)$. Spreads and partial spreads are basic concepts
which were very well studied in projective geometry.
We will not go into all the details, except for some concepts related to
spreads which will be discussed in the sections which follow.
We already mentioned that a spread in $\cG_q(n,k)$, where $k$
divides $n$ is a Steiner structure $\dS_q(1,k,n)$.
The value of $\cA_q(n,k,k)$ is of a very special interest.
This value has a special interest since $(n,k,k)_q$ codes have
applications as byte-correcting codes~\cite{Etz98,HoPa72}.
Decoding of such constant dimension codes was considered in~\cite{GMR12,MGR08}

The known upper and lower bounds on $\cA_q(n,k,k)$
are summarized in the following theorems. The first three well-known
theorems can be found in~\cite{EtVa11}.

\begin{theorem}
\label{thm:exactk_2k}
If $k$ divides $n$ then $\cA_q(n,k,k) = \frac{q^n-1}{q^k-1}$.
\end{theorem}

\begin{theorem}
\label{thm:upperk_2k}
$\cA_q (n,k,k) \leq \left\lfloor \frac{q^n-1}{q^k-1} \right\rfloor -1$
if $n \not\equiv 0~(\text{mod}~k)$.
\end{theorem}

\begin{theorem}
\label{thm:lowerk_2k}
Let $n \equiv r~(\text{mod}~k)$. Then, for all
$q$, we have
$$
\cA_q(n,k,k) \geq  \frac{q^n - q^k(q^r -1)-1}{q^k-1}
$$
\end{theorem}

We note that one method to obtain the lower bound
of Theorem~\ref{thm:lowerk_2k} is to apply the Multilevel Construction.
The next theorem was proved in~\cite{HoPa72} for $q=2$ and for
any other $q$ in~\cite{Beu75}.

\begin{theorem}
\label{thm:S_exactk_2k}
If $n \equiv 1~(\text{mod}~k)$ then $\cA_q(n,k,k) = \frac{q^n-q}{q^k-1} -q+1=\sum_{i=1}^{\frac{n-1}{k}-1} q^{ik+1}+1$.
\end{theorem}

The bound of Theorem~\ref{thm:S_exactk_2k} is attained with $\CMRD$
to which one subspace is added.
By Theorems~\ref{thm:exactk_2k} and~\ref{thm:S_exactk_2k}, the
value of $\cA_q(n,2,2)$ is known for all values of $q$ and $n$.
Theorem~\ref{thm:S_exactk_2k} was extended for the case where
$q=2$ and $k=3$ in~\cite{EJSSS} as follows.

\begin{theorem}
If $n \equiv c~(\text{mod}~3)$ then $\cA_2(n,3,3) = \frac{2^n-2^c}{7} -c$.
\end{theorem}

The upper bound implied by
Theorem~\ref{thm:S_exactk_2k} was improved for some cases
in~\cite{DrFr79} in which a transformation, of partial spreads into
orthogonal arrays of strength two, is considered.

\begin{theorem}
\label{thm:best_U_bound}
If $n=k \ell + c$ with $0 < c < k$, then
$\cA_q(n,k,k) \leq \sum_{i=0}^{\ell-1} q^{ik+c} -\Omega -1$,
where $2 \Omega = \sqrt{1+4q^k(q^k-q^c)} -(2q^k-2q^c+1)$.
\end{theorem}

\begin{problem}
Improve the lower bound on $\cA_q(n,k,k)$ given in Theorem~\ref{thm:lowerk_2k}.
\end{problem}

\begin{problem}
Characterize the cases in which the lower bound
on $\cA_q(n,k,k)$ given in Theorem~\ref{thm:lowerk_2k} is
the exact value of $\cA_q(n,k,k)$.
\end{problem}

\begin{problem}
Improve the upper bound on $\cA_q(n,k,k)$ given in Theorem~\ref{thm:best_U_bound}.
\end{problem}

\begin{problem}
Find more parameters for which we can give the exact value of $\cA_q(n,k,k)$
as given in Theorem~\ref{thm:S_exactk_2k}.
\end{problem}

\begin{problem}
Find the value of $\cA_q(n,3,3)$ for all $q$ and $n$.
\end{problem}

\begin{problem}
Find the value of $\cA_2(n,4,4)$ for a new infinite family
of values of~$n$.
\end{problem}

\section{Rank-Metric Codes}
\label{sec:rank}

Ferrers diagram rank-metric codes are the key for large codes based
on the Multilevel Construction.
Therefore, it is important (at least theoretically) to find large
Ferrers diagram rank-metric codes. Let $\dim (\cF , \delta )$ be the
the largest possible dimension of an $[\cF,\varrho,\delta]$ code.
The following theorem for the upper bound on the size of such codes
was proved in~\cite{EtSi09}.

\begin{theorem}
\label{thm:upper_rank} For a given $i$, $0 \leq i \leq \delta -1$,
if $\nu_i$ is the number of dots in
a Ferrers diagram $\cF$, which are not contained
in the first $i$ rows and are not contained in the rightmost
$\delta-1-i$ columns then $\text{min}_i \{ \nu_i \}$ is an upper
bound of $\dim (\cF,\delta)$.
\end{theorem}

The bound of Theorem~\ref{thm:upper_rank} is attained trivially for $\delta=1$
and also attained for $\delta=2$ and for some sporadic parameters~\cite{EtSi09}.
Therefore, the most important questions in this context are as follows.

\begin{problem}
Find more parameters for which the bound of Theorem~\ref{thm:upper_rank} is attained.
\end{problem}

\begin{problem}
Prove that the bound of Theorem~\ref{thm:upper_rank} is attained for all sets of parameters,
$\cF$, $q$, and $\delta$; or disprove this claim.
\end{problem}

Rank-metric codes can raise some more interesting questions, but usually
these questions do not have a direct application or any relation for subspace codes.
However, a specific type of rank-metric codes, namely \emph{constant rank codes}
have an important direct application to constant dimension codes.
These codes were considered in~\cite{GaYa10a}. A constant rank code
is a rank-metric code in which all codewords have the same rank.

Let $\cA_q^R(m,n,d,r)$ be the maximum number of codewords in
a constant rank code with constant rank $r$ and
minimum rank distance $d$ over $\F_q^{n \times m}$.
The following two theorems~\cite{GaYa10a}
are the most relevant ones in the context
of optimal constant dimension codes.

\begin{theorem}
For all $q$, $2k \leq n \leq m$ and $1 \leq \delta \leq r$ we have
$\cA_q(n,\delta, r)=\cA_q^R (m,n,\delta +r,r)$ if either $\delta =r$
or $m \geq (n-r)(r-\delta+1)+r+1$.
\end{theorem}

\begin{theorem}
$~$
\begin{enumerate}
\item $\cA_q^R (m,n,r+1,r)= \sbinomq{n}{r}$.

\item $\cA_q^R (n,m,2r,r)= \cA_q(n,r,r)$ (the largest partial spread).
\end{enumerate}
\end{theorem}

\begin{problem}
Find constructions for constant rank codes which are not derived
from the known constructions of constant dimension codes.
\end{problem}

\begin{problem}
Continue to develop the theory of constant rank codes beyond the results
given in~\cite{GaYa10a}.
\end{problem}

\begin{problem}
Find the exact value of $\cA_q^R (n,m, \delta +r ,r)$ for
$2 \leq \delta < r$, where $m \geq (n-r)(r-\delta+1)+r+1$.
\end{problem}

\begin{problem}
Develop the theory for $\cA_q^R (n,m, \delta +r ,r)$ for
$2 \leq \delta < r$, where $m < (n-r)(r-\delta+1)+r+1$.
\end{problem}

\section{Decoding of Subspace Codes}
\label{sec:decoding}

Error-correcting codes over any channel are constructed
for the purpose of correcting errors caused during the transmission
of the information on the channel. Therefore,
from a theoretical point of view one might be interested
in the size of the largest possible code. But,
from a practical point of view when a code
is constructed, we are more interested in its
decoding (but not neglecting the requirement for a large
code). For some of the first constructions
mentioned earlier the authors gave decoding algorithms,
e.g.~\cite{EtSi09,GMR12,KoKs08,MGR08,SKK08,Ska10}.
All these decoding algorithms are based on maximum likelihood decoding.
Most of these decoding algorithms are for
Grassmannian codes, but some can be adopted
for subspace codes with either the subspace
distance or the injection distance.
Quite naturally also list-decoding algorithms were developed
for Grassmannian codes (constructed either by linearlized polynomials or
as lifted rank-metric codes),
e.g~\cite{Aga11,GNW12,GuXi12,MaVa10,MaVa11,MaVa12a,MaVa12,RoTr12,TSR13}.
Also, this direction of research has many problems for future research.

\begin{problem}
Suggest new classes of large Grassmannian codes, subspace codes
(with the subspace distance or the injection distance)
with efficient decoding algorithms.
\end{problem}

\begin{problem}
Design a list-decoding algorithm for codes which are not subcodes
of either linearlized polynomial codes or lifted rank-metric codes.
\end{problem}

\begin{problem}
Find good lower and upper bounds on the size of Grassmannian
codes for list-decoding when the size of the list is a small constant.
\end{problem}

\begin{problem}
Find good lower and upper bounds on the size of subspace
codes (not necessarily Grassmannian)
for list-decoding when the size of the list is a small constant.
\end{problem}

The subspace codes constructed by the various construction methods
are not linear and therefore we should have an encoding algorithm
from the list of information words into the list of codewords.
If the code consists for example from one lifted MRD
code then the encoding is trivial by using the
encoding of the related rank-metric code. If the code is
constructed by the Multilevel Construction then the
encoding is slightly more complicated and it was described
in~\cite{KSK09}. Finally, encoding of all the Grassmannian space was
described in~\cite{Med12,SiEt11}.

\begin{problem}
Find better encoding algorithms for subspace codes, the
Grassmannian $\cG_q(n,k)$, or the projective space $\cP_q(n)$.
\end{problem}

\section{Designs over GF($q$)}
\label{sec:designs}

$q$-analog of Steiner systems are one type of $q$-analog
of designs, also called designs over~$\F_q$.
As was mentioned before,
the notion of $t$-design have been extended to
vector spaces by Cameron~\cite{Cam74,Cam74a} and Delsarte~\cite{Del76}
in the early 1970s.

A $t-(n,k,\lambda)_q$ design is a collection $\B$
of $k$-dimensional subspaces (called \emph{blocks}) from $\cG_q(n,k)$ such that
each $t$-dimensional subspace of $\cG_q(n,t)$ is contained in
exactly $\lambda$ blocks of $\B$. If $\B$ contains all the
$k$-dimensional subspaces of $\cG_q(n,k)$ then the design
is said to be trivial.

Thomas~\cite{Tho87} was the first to find nontrivial
$t$-design over $\F_q$, which are not spreads.
The work of Thomas has motivated other research work to explore
this topic and more $t$-designs over $\F_q$
were found~\cite{BKL05,Ito98,MMY95,Suz90,Suz92,Tho96}.

Another type of design over $\F_q$ which was defined
is the \emph{subspace transversal design}~\cite{EtSi13}.
It is not a direct $q$-analog
of a transversal design as will be explained in the sequel.

Let $\V^{(n,k)}$ be the set of nonzero vectors of $\F_q^n$ whose
first $k$ entries form a nonzero vector.
For a given $X \in \cG_q(k,1)$, let $\V_X^{(n,k)}$ denote the set nonzero
vectors in $\F_q^n$ whose first $k$ entries form any given nonzero vector of $X$.
Let $\V_{\bf 0}^{(n,k)}$ denote a maximal set of $\frac{q^{n-k}-1}{q-1}$ nonzero
vectors in $\F_q^n$ whose first $k$ entries are \emph{zeroes}, for which
any two vectors in the set are linearly independent.
Let $\V_{\bf 0}$ denote the $k$-dimensional subspace spanned by $\V_{\bf 0}^{(n,k)}$.

A \emph{subspace transversal  design} of groupsize $q^{n-k}$,
block dimension $k$, and \emph{strength}~$t$, denoted by
$\text{STD}_q (t, k, n-k)$, is a triple
$(\V,\mathbb{G},\mathbb{B})$, where
$\V$ is a set of points,
$\mathbb{G}$ is a set of groups, and $\mathbb{B}$ is
a set of blocks. These three sets must satisfy the following five
properties:

\begin{enumerate}
\item $\V$ is a set of size $\frac{q^k-1}{q-1} q^{n-k}$ (the \emph{points}).
$\bigcup_{X \in \cG_q(k,1)} \V_X^{(n,k)}$ is used as the set of points $\V$.

\item $\mathbb{G}$ is a partition of $\V$ into
$\frac{q^k-1}{q-1}$ classes of size $q^{n-k}$ (the \emph{groups});
the groups which are used are defined by
$\V_X^{(n,k)}$, $X \in \cG_q(k,1)$.

\item $\mathbb{B}$ is a collection of $k$-dimensional
subspaces of $\F_q^n$ which contain nonzero vectors
only from $\mathbb{V}^{(n,k)}$ (the
\emph{blocks});

\item each block meets each group in exactly one point;

\item every $t$-dimensional subspace (with points from $\mathbb{V}$) which
meets each group in at most one point is contained in exactly
one block.
\end{enumerate}

This is not a direct $q$-analog of a transversal design since
the elements of $\V_{\bf 0}^{(n,k)}$
don't participate in any block of the
design. It was proved in~\cite{EtSi13} that
the codewords of an $(n,\delta,k)_q$
$\CMRD$ form the blocks of a
$\text{STD}_{q}(k-\delta+1,k, n-k)$. It was also
shown in~\cite{EtSi13} how to use the properties of
subspace transversal design to obtain better bounds
on $\cA_q(n,\delta,k)$ with codes which contains $\CMRD$.
These properties were also used to construct $q$-covering
designs~\cite{Etz13} and parallelisms~\cite{Etz13a} and they probably can
be used for constructions of other related structures.

\begin{problem}
Find new $t-(n,k,\lambda)_q$ designs with new parameters.
\end{problem}

\begin{problem}
Find $q$-analogs for other known types of block designs
for which an application for a construction of subspace
codes can be given.
\end{problem}

\begin{problem}
Find $q$-analogs for other types of combinatorial designs,
such as Latin squares, orthogonal arrays, etc.
\end{problem}

\begin{problem}
Prove that for each $1 < t < k$ and each $q$ there exists an integer $n_0 > k$
such that for each $n > n_0$ a nontrivial $t-(n,k,\lambda)_q$ design exists.
\end{problem}

Recently, another type of $q$-analog for designs was considered.
This is a large set of a $t-(n,k,\lambda)_q$ design.
A \emph{large set} of a design $\cS$ is a partition of the
space into disjoint copies of $\cS$. Hence, a large set of
$t-(n,k,\lambda)_q$ designs is a partition of $\cG_q(n,k)$
into disjoint copies of $t-(n,k,\lambda)_q$ designs.
Parallelism in projective geometry is a large set
and this topic will be discussed separately in Section~\ref{sec:parallel}.
Braun, Kohnert, \"Osterg\aa rd, and Wassermann~\cite{BKOW}
presented a large set of $2-(8,3,21)_2$ designs.
This large set consists of three
disjoint $2-(8,3,21)_2$ designs.

\begin{problem}
Find more large sets of $t-(n,k,\lambda)_q$ designs.
\end{problem}

\section{$q$-Covering Designs}
\label{sec:covering}

A \emph{$q$-covering design} $\C_q(n,k,r)$ is a collection $\dS$ of
elements from $\cG_q(n,k)$ such that each element of $\cG_q(n,r)$
is contained in at least one element of $\dS$.
Let $\cC_q(n,k,r)$ denote the minimum number of subspaces in a
$q$-covering design $\C_q(n,k,r)$.

$q$-covering designs were considered first in the context of
projective geometry. A set $\T$ of $t$-subspaces in PG($n,q$) such that each $s$-subspace
contains at least one element of $\T$ is called a \emph{blocking set}.
Such a design is a $q$-analog of the well-known
{\it Tur\'{a}n design}~\cite{Caen83,Caen91}. The
dual subspaces of the subspaces in a
blocking set form a $q$-covering design $\C_q(n+1,n-t,n-s)$.
Blocking sets were considered for example in~\cite{Met03,Met04}.
We note that blocking sets have also some different definitions
(and maybe more popular definitions which define
other structures which are not $q$-coverings).

Similarly, to the case of error-correcting codes in the
projective space (which are $q$-packing designs) there are some basic bounds
on the size of a $q$-covering design. The first one is
the $q$-analog Sch\"{o}nheim bound~\cite{Sch64} which was given
in~\cite{EtVa11a}.

\begin{theorem}
\label{thm:schonheim} $\cC_q(n,k,r) \geq \left\lceil
\frac{q^n-1}{q^k-1} \cC_q(n-1,k-1,r-1) \right\rceil$.
\end{theorem}

The basic covering bound is given in the following theorem~\cite{EtVa11a}.
\begin{theorem}
\label{thm:covering_bound} $\cC_q(n,k,r) \geq
\begin{small} \frac{\sbinomq{n}{r}}{\sbinomq{k}{r}} \end{small}$
with equality holds if and
only if a Steiner structure $\dS_q(r,k,n)$ exists.
\end{theorem}

As in the case of the $q$-analog Johnson bound (Theorem~\ref{thm:packing}) also the
$q$-analog Sch\"{o}nheim bound can be iterated.

\begin{theorem}
$$
\cC_q(n,k,r) \geq \left\lceil\frac{q^n\!-\!1}{q^k\!-\!1}
\left\lceil\frac{q^{n-1}\!-\!1}{q^{k-1}\!-\!1} \cdots
\left\lceil\frac{q^{n-r+1}\!-\!1}{q^{k-r+1}\!-\!1}\right\rceil
\cdots\right\rceil\right\rceil \geq
\frac{\sbinomq{n}{r}}{\sbinomq{k}{r}}~.
$$
\end{theorem}

Another lower bound given in~\cite{EtVa11a} is a $q$-analog of a
theorem given by de Caen in~\cite{Caen83,Caen91}.

\begin{theorem}
\label{thm:bound_k,k-1} $\cC_q(n,k,k-1) \geq
\frac{(q^k-1)(q-1)}{(q^{n-k}-1)^2} \sbinomq{n}{k+1}$.
\end{theorem}

The following two theorems given for example
in~\cite{EtVa11a} are also simple to obtain.
\begin{theorem}
\label{thm:optc2} If $1 \leq k \leq n$, then
$\cC_q(n,k,1)=\left\lceil \frac{q^n-1}{q^k-1} \right\rceil$.
\end{theorem}

\begin{theorem}
\label{thm:Turan=1} If $1 \leq r \leq n-1$, then $\cC_q(n,n-1,r) = \frac{q^{r+1}-1}{q-1}$.
\end{theorem}

Theorem~\ref{thm:Turan=1} was proved before in the context of
projective geometry by Bose and Burton in~\cite{BoBu66}.
Another lower bound was given in~\cite{EiMe97} by considering sets
of lines in PG($2s,q$) contained in $s$-subspaces.
\begin{theorem}
\label{thm:EiMe}
$\cC_q (2s+1,2s-1,s) \geq \frac{q^{2s+2}-q^2}{q^2-1} + \frac{q^{s+1}-1}{q-1}$ for
every integer $s \geq 2$.
\end{theorem}

Metsch~\cite{Met03} also gave a construction for a set of lines in
PG($2s+x-1,q$), for every $1 \leq x \leq s$, contained in $s$-subspaces,
which yields the following theorem:
\begin{theorem}
\label{thm:Met03a}
For any given integers $q \geq 2$, $1 \leq x \leq s$ we have
$\cC_q (2s+x,2s+x-2,s+x-1)  \leq \frac{q^{2s+2x}-q^{2x}}{q^2-1}
+ \frac{q^x-1}{q-1} \cdot \frac{q^{s+x}-q^{x-1}}{q-1}$.
\end{theorem}

Finally, also Theorem~\ref{thm:optc2} was proved in terms of
projective geometry. In projective geometry, the quantity $\cC_q(n+1,k+1,1)$,
is the minimum number of $k$-subspaces in PG($n,q$) such that
each point of PG($n,q$) is contained in at least one of these subspaces.
The solution obtained in Theorem~\ref{thm:optc2} was obtained
before by Beutelspacher~\cite{Beu79}.
The proofs for all these results and related results
in projective geometry were also given by Metsch in~\cite{Met03}.

The most basic upper bound, given by construction, on the size of a $q$-covering
design was proved in~\cite{EtVa11a}.
\begin{theorem}
\label{thm:recursiveC}
$\cC_q (n,k,r) \leq q^{n-k} \cC_q
(n-1,k-1,r-1) + \cC_q (n-1,k,r)$.
\end{theorem}

Normal spreads~\cite{Lun99}, also known as geometric spreads~\cite{BeUe91},
are used to prove the following values of $\cC_q(n,k,r)$~\cite{BlEt11}.

\begin{theorem}
\label{thm:n_spreads}
$\cC_q (vm+\delta,vm-m+\delta,v-1) =
\frac{q^{vm}-1}{q^m-1}$ for all $v \geq 2$, $m \geq 2$, and $\delta \geq 0$.
\end{theorem}

The next theorem given in~\cite{EtVa11a} is used infinitely many times once
an exact bound for some given parameters is known.
\begin{theorem}
\label{thm:lengthening} $\cC_q (n+ 1 , k + 1 , r) \leq \cC_q (n,k,r)$.
\end{theorem}
Theorem~\ref{thm:lengthening} implies a very interesting property on the behavior
of optimal $q$-design coverings.
\begin{cor}
\label{cor:leng} For any given $r >0$ and $\delta >0$ there exists
a constant $c_{q,\delta,r}$ and an integer $n_0$ such that for
each $n > n_0$, $\cC_q(n,n-\delta,r)=c_{q,\delta,r}$.
\end{cor}

The usage of lifted MRD codes as subspace transversal designs made
it possible to obtain some interesting bounds on $\cC_2(n,k,2)$ and
$\cC_2(n,k,3)$~\cite{Etz13}.

\begin{problem}
Find new techniques to construct $q$-covering designs for $q=2$.
\end{problem}

\begin{problem}
Find new techniques to obtain new lower bounds
on $\cC_2(n,k,r)$.
\end{problem}

\begin{problem}
Improve the bounds on $\cC_2(n,k,r)$ for $n \leq 10$ (tables for the
known bounds are given in~\cite{Etz13}).
\end{problem}

\begin{problem}
Find new parameters for which the exact value of $\cC_2(n,k,r)$ can be obtained.
\end{problem}

\begin{problem}
Develop new techniques to construct $q$-covering designs for $q>2$
and generate related tables for the bounds on $\cC_q(n,k,r)$
for small $q$, e.g. $q=3$, 4, and 5.
\end{problem}

\begin{problem}
Make use of subspace transversal designs to obtain new bounds
on $\cC_q (n,k,r)$ for $q > 2$ or for $q=2$ and $r > 3$.
\end{problem}

Another type of codes which can be considered are covering codes
in the projective space which are the $q$-analog for covering
codes in the Hamming scheme. A code $\C$ is a $q$-covering code
with covering radius $R$ in $\cP_q(n)$ if $\C$ consists of subspaces
from $\cP_q(q)$ and for each subspace $X \in \cP_q(n)$ there
exists a subspace $Y \in \C$ such that $d_S(X,Y) \leq R$.
Let $\cC_q(n,R)$ be the minimum size
of a $q$-covering code with radius $R$ in $\cP_q(n)$. To develop
the theory of $q$-covering codes in $\cP_q(n)$ we can consider the
following problems.

\begin{problem}
Develop some basic lower bounds on $\cC_q(n,R)$.
\end{problem}

\begin{problem}
Develop constructions to obtain upper bounds on $\cC_q(n,R)$.
\end{problem}

\begin{problem}
Compile a table with lower and upper bounds on $\cC_q(n,R)$ for
small values of $q$ and $n$.
\end{problem}

Some work in this direction was done in~\cite{GaYa10}. The bounds on
$\cC_q(n,k,r)$ can play an important role in this direction (to obtain
bounds on $\cC_q(n,R)$)
in the same way that the bounds on $\cA_q(n,\delta,k)$ play an important role
when the bounds on $\cA_q^S (n,d)$ are considered. These covering questions can be
considered similarly for the injection distance instead of the subspace distance.
Covering codes can be considered also for the injection distance~\cite{GaYa10}.
But, $q$-covering codes are mainly interesting from
a theoretical point of view. From combinatorial perspective
the subspace distance is more interesting than the injection distance
and hence covering codes with the injection
distance might attract less attention.

\section{Asymptotic Behavior}
\label{sec:asymptotic}

One important topic to consider is the asymptotic behavior
of $\cA_q(n,d,k)$, $\cA_q^S (n,d)$, $\cA_q^I (n,d)$,
$\cC_q(n,k,r)$, and $\cC_q(n,R)$. A family of Grassmannian codes $\C_n$,
of subspaces from $\cG_q(n,k)$ with minimum Grassmannian distance $\delta$,
is called \emph{asymptotically optimal} if
$\frac{| \C_n |}{\cA_q(n,\delta,k)} \rightarrow 1$ as $n \rightarrow \infty$.
The quantity $\frac{| \C_n |}{\cA_q(n,\delta,k)}$
as $n \rightarrow \infty$ is called the \emph{density}
of the code.
Similarly, we define asymptotically optimal subspace codes
with either the subspace distance or the
injection distance, and $q$-covering designs. Density for the other
types of codes is also defined similarly. It is quite obvious that
there could be many cases in which we are not able to know whether
the codes are asymptotically optimal since we don't know the magnitude
of the optimal codes. Therefore, in the context of the asymptotic
behavior we are concerned with the following three problems:
\begin{enumerate}
\item What is the size of an asymptotically optimal code?

\item What are the lower bounds on the density of the known codes
(upper bounds for covering codes and designs)?

\item Constructions of asymptotically optimal codes.
\end{enumerate}

Of course, we are mostly interested for
an answer to the third problem since constructions
of asymptotically optimal codes also determine the size of asymptotically
optimal codes and produces codes with density one.

\begin{problem}
Find constructions for families of asymptotically optimal Grassmannian
codes.
\end{problem}

\begin{problem}
Find constructions for families of asymptotically optimal
$q$-covering designs.
\end{problem}

But, when a construction
of such asymptotically optimal
codes is not available we would like to consider solutions for the
the other two problems.

Blackburn and Etzion~\cite{BlEt11} consider the asymptotic behavior
of Grassmannian codes and $q$-covering designs. They first presented
a connection between the size of an optimal Grassmannian code and
the size of an optimal $q$-covering design.

\begin{theorem}
\label{thm:design_equiv_code}
We have that
$$
\cC_q(n,k,k-\delta)\leq
\cA_q(n,\delta+1,k)+\\
\sbinomq{n}{k-\delta}-\sbinomq{k}{k-\delta}
\cA_q(n,\delta+1,k)
$$
and
$$
\cA_q(n,\delta+1,k)\geq \cC_q(n,k,k-\delta) +\\
\sbinomq{n}{k-\delta}-
\sbinomq{k}{k-\delta}\cC_q(n,k,k-\delta).
$$
\end{theorem}

Let $A(n) \sim B(n)$ means that $\lim_{n \rightarrow \infty} A(n)/B(n)=1$.
By considering hypergraphs, probabilistic arguments, and Theorem~\ref{thm:design_equiv_code}
we have the following two theorems~\cite{BlEt11}.
\begin{theorem}
\label{thm:packing_asymptotics} Let $q$, $k$ and $\delta$ be fixed
integers, with $0 \leq \delta\leq k$ and such that $q$ is a prime
power. Then
\begin{equation}
\label{eqn_packing}
\cA_q(n,\delta+1,k)\sim \frac{\sbinomq{n}{k-\delta}}{\sbinomq{k}{k-\delta}}
\end{equation}
as $n\rightarrow\infty$.
\end{theorem}
\begin{theorem}
\label{thm:covering_asymptotics} Let $q$, $k$ and $\delta$ be
fixed integers, with ${0\leq \delta\leq k}$ and such that $q$ is a
prime power. Then
\[
\cC_q(n,k,k-\delta)\sim \frac{\sbinomq{n}{k-\delta}}{\sbinomq{k}{k-\delta}}
\]
as $n\rightarrow\infty$.
\end{theorem}
For specific values, asymptotically optimal
Grassmannian codes and $q$-covering designs
were mentioned in the previous sections. But, the following problems
remained unsolved for most parameters.

\begin{problem}
Extend the range for which the size of
asymptotically optimal Grassmannian codes
can be shown.
\end{problem}

\begin{problem}
What is the asymptotic behavior of $\cA_q^S (n,d)$ and $\cA_q^I (n,d)$?
\end{problem}

\begin{problem}
What is the asymptotic behavior of $\cA_q (n,\delta,k)$ when $\delta$ or $k$
are not fixed?
\end{problem}

\begin{problem}
What is the asymptotic behavior of $\cC_q (n,k,r)$ when $k$ or $r$
are not fixed?
\end{problem}

\begin{problem}
What is the asymptotic behavior of $\cC_q (n,R)$?
\end{problem}

When we cannot construct asymptotically optimal codes
we are interested in dense Grassmannian codes and sparse
$q$-covering designs. Koetter and Kschischang~\cite{KoKs08}
proved that the density of $\CMRD$ codes is at least $\frac{1}{4}$.
Thus, this class of codes is good enough for any practical purpose.
The lower bounds on the density were considerably improved with
better constructions and improved upper bounds on the size of
the codes. An analysis of the lower bounds on the densities is
given in~\cite{EtSi13}, where it was shown that the density
is at least $\frac{3}{5}$. Some work in this direction for
$q$-covering designs was done in~\cite{Etz13}. But, there is a lot of
ground for further research in this direction.

\begin{problem}
Show a general bound of considerably more than $\frac{3}{5}$
for the density of Grassmannian codes.
\end{problem}

\begin{problem}
Produce a comprehensive analysis of the densities for
Grassmannian codes and $q$-covering designs.
\end{problem}

\section{Parallelism}
\label{sec:parallel}

A $k$-spread is called a parallel class as it partition
the set of all the points of PG($n,q$). A $k$-parallelism in
PG($n,q$) is a partition of the $k$-subspaces
of PG($n,q$) into pairwise disjoint $k$-spreads.
Hence, a parallelism is also a type of a large set
as mentioned in Section~\ref{sec:designs}. Some 1-parallelisms
of PG($n,q$) are known for many years. For $q=2$ and
odd~$n$ there is an 1-parallelism in PG($n,2$). Such a parallelism
was found in the context of Preparata codes and it is known that
many such parallelisms exist~\cite{Bak76,BLW,ZZS71}. For any other power
of a prime~$q$, if $n=2^i-1$, $i \geq 2$, then
an 1-parallelism was shown in~\cite{Beu74}.
In the last forty years no new parameters for
1-parallelisms were shown until recently, when
an 1-parallelism in PG(5,3) was proved to exist in~\cite{EtVa12}.
A $k$-parallelism for $k > 1$ was not known until a 2-parallelism
in PG($5,2$) was shown in~\cite{Sar02}.

Clearly, such parallelisms can be described in terms of spreads
in the Grassmannian. As it seems to be extremely difficult,
we consider two problems which are generalizations of the
parallelism problem. The first one is to consider what
is the maximum number
of pairwise disjoint $k$-spreads that exist in PG($n,q$)?
Beutelspacher~\cite{Beu90} has proved that if $n$ is odd then there exist
$q^{2 \lfloor \log n \rfloor} + \cdots +q+1$ pairwise
disjoint 1-spreads in PG($n,q$). It is proved in~\cite{Etz13a}
that two disjoint $k$-spreads exist
in PG($n,q$) and $2^{k+1}-1$ pairwise disjoint spreads
exist in PG($n,2$).

For the second problem,
we will define a \emph{partial Grassmannian} $\cG_q(n_1,n_2,k)$,
$n_1 > n_2 \geq k$, as the set of all $k$-dimensional subspaces from
the space $\F_q^{n_1}$
which are not contained in a given $n_2$-dimensional subspace $U$ of $\F_q^{n_1}$.
It can be readily verified that $\V^{(n,k)}$ is a partial Grassmannian
$\cG_q(n,n-k,k)$, where $\V_0^{(n,k)}$ is the $(n-k)$-dimensional
subspace $U$.
A spread in $\cG_q(n_1,n_2,k)$ is a set $\dS$ of pairwise disjoint $k$-dimensional
subspaces from $\cG_q(n_1,n_2,k)$ such that each nonzero element of $\F_q^{n_1} \setminus U$
is contained in exactly one element of~$\dS$. A parallelism of $\cG_q(n_1,n_2,k)$
is a set of pairwise disjoint spreads in $\cG_q(n_1,n_2,k)$ such that each $k$-dimensional
subspace of $\cG_q(n_1,n_2,k)$ is contained in exactly one of the spreads.
Beutelspacher~\cite{Beu90} proved that if $k=2$ then such a parallelism exists
if $n_2 \geq 2$, $n_1 - n_2 = 2^i$, for all $i \geq 1$ and any~$q > 2$.
If $k=2$ and $q=2$ then such a parallelism exists
if and only if $n_2 \geq 3$ and $n_1 - n_2$ is even.
Etzion~\cite{Etz13a} proved that
if $k=n_1 - n_2$ then there exists a parallelism in $\cG_q(n_1,n_2,k)$

\begin{problem}
For any $q >2$ and $k \geq 1$, improve the lower bounds on the number of pairwise disjoint
$k$-spreads in PG($n,q$).
\end{problem}

\begin{problem}
For $q =2$ and any $k > 1$, improve the lower bounds on the number of pairwise disjoint
$k$-spreads in PG($n,q$).
\end{problem}

\begin{problem}
Find nontrivial necessary conditions for the existence of a parallelism in
$\cG_q(n_1,n_2,k)$.
\end{problem}

\begin{problem}
Find new parameters for which there exists a parallelism in
$\cG_q(n_1,n_2,k)$.
\end{problem}

\begin{problem}
For a power of a prime $q>2$, find new parameters for which there
exists an 1-parallelism in PG($n,q$).
\end{problem}

\begin{problem}
For $k>1$ and any power of a prime $q$, starting
with $q=2$, find new parameters for which there
exists a $k$-parallelism in PG($n,q$).
\end{problem}

\begin{problem}
For $k>1$, find an infinite family of $k$-parallelisms in PG($n,q$).
\end{problem}

\section{Other Problems in Coding Theory}
\label{sec:other}

There are many other interesting problems in connections
to coding theory which can be defined on the projective space
and the Grassmannian. We will consider three topics: Gray codes,
self-complements codes, and linear codes.

\subsection{Gray Codes}

Gray codes have many applications and they are defined on
variety of objects~\cite{Sav97}.
A Gray code in $\cP_q(n)$ or $\cG_q(n,k)$ is a path in the
related graphs $G(\cP_q^S(n))$ and
$G(\cG_q(n,k))$, respectively. In $G(\cP_q^S(n))$
the vertices represent the subspaces of $\F_q^n$. Two vertices $X$ and $Y$
are connected by an undirected edge if $d_S (X,Y) =1$.
In $G(\cG_q(n,k))$
the vertices represent the $k$-dimensional subspaces
of $\F_q^n$. Two vertices $X$ and $Y$
are connected by an undirected edge if $d_S (X,Y) =2$ ($d_G (X,Y) =1)$.
The goal is to find the longest path and if possible a path
which contains all the vertices in the graph, i.e. it will be
a Hamiltonian path. Moreover, it is
also desired that the path will be a cycle. One can easily verify
that in $\cP_q(n)$ there is no Hamiltonian path if $n$ is even since the
number of vertices with even dimension is greater than the
number of vertices with odd dimension. Therefore, the four obvious
research problems can be stated as follows.

\begin{problem}
Is there a Hamiltonian cycle in $G(\cP_q^S(n))$, $n$ odd?
\end{problem}

\begin{problem}
What is the the length of the longest path in $G(\cP_q^S(n))$?
\end{problem}

\begin{problem}
What is the the length of the longest cycle in $G(\cP_q^S(n))$?
\end{problem}

\begin{problem}
Is there a Hamiltonian cycle in $G(\cG_q(n,k))$?
\end{problem}

\begin{problem}
What is the the length of the longest path in $G(\cG_q(n,k))$?
\end{problem}

We note that the graph $G(\cP_q^I(n))$ is not of interest for
this problem as it lacks the natural combinatorial structure
as the other two graphs.

Another interesting problem in this context
is the $q$-analog of the middle levels
problem which is a well-known
unsolved problem for the Hamming graph~\cite{SaWi95}.
The $q$-analog problems are presented as follows.

\begin{problem}
Is there a cycle in $G(\cP_q^S(2k+1))$ which contains
all the $k$-dimensional subspaces and all
the $(k+1)$-dimensional subspaces?
\end{problem}

\begin{problem}
What is the length of the longest path in $G(\cP_q^S(2k+1))$ which contains
only $k$-dimensional subspaces and $(k+1)$-dimensional subspaces?
\end{problem}

In~\cite{Etz13b} it is shown that for any given $q$ and $k=1$ or $k=2$
there exists a Hamiltonian cycle in the
middle levels of $\cP_q(2k+1)$. The method is using cyclic shifts
of subspaces in a modification of a similar method which is making use of
necklaces in the Hamming graph.

\subsection{Complements}

Complements of binary codewords and binary codes are used as a tool
in various aspects of coding theory. The $q$-analog was considered
in~\cite{BEV13}. Various related problems concerning
complements of subspaces over $\F_q$ were considered before,
e.\,g.~\cite{Cla92,ClSh95,ClSh97}.

\begin{defn}
Let $\cU$ be a subset of $\cP_q(n)$ and let $\cU_k:=\cU \cap \cG_q(n,k)$.
We say that a function $f: \cU
\to \cU$ is a \emph{complement on $\cU$} (and denote
$\overline{X} = f(X)$ for all $X\in\cU$) if $f$ has the following
properties:

\begin{enumerate}
\item[\bf P1.] $X \cap \overline{X} = \{ {\bf 0} \}$ and $X + \overline{X} =
\F_q^n$, i.e. $X\oplus \overline{X} =\F_q^n$ for all $X \in \cU$.

\item[\bf P2.] $f$ establishes a bijection between $\cU_k$ and $\cU_{n-k}$
for all $k$, $0\le k\le n$.

\item[\bf P3.] $f(f(X)) = X$ for all $X \in \cU$.

\item[\bf P4.] $d_S(\overline{X},\overline{Y}) = d_S(X,Y)$ for all $X,Y \in
\cU$.
\end{enumerate}
\end{defn}

The existence problems of complements in $\cP_q(n)$ was
considered in~\cite{BEV13}. Some of the results are based on
representation of subspaces by lattices~\cite{Bru11}.
The main open problem which remains unsolved in this discussion
is our next open problem.

\begin{problem}
Prove that the largest subset of $\cP_q(n)$ on which a complement
can be defined is the set $\cV_q(n)=\{X\in\cP_q(n):X\cap X^\perp = \{ {\bf 0} \} \}$;
or disprove this claim.
\end{problem}

A closed-form expression for $|\cV_q(n)|$ was given by
Sendrier~\cite{Sen97}. Using the results of~\cite{Sen97}, it
can be shown that the size of $\cV_q(n)$ is \emph{proportional} to
$|\cP_q(n)|$, specifically:
\[
\lim_{n\to\infty}\frac{|\cV_q(n)|}{|\cP_q(n)|} = \prod_{i=1}^\infty
\frac{1}{1+q^{-i}}~.
\]
The limit converge to $0.4194 \ldots$ when $q=2$, $0.639 \ldots$
when $q=3$, $0.7375 \ldots$ when $q=4$, and $0.9961 \ldots$ when
$q=256$.

\subsection{Linear Codes}

Linear codes are one of the basic concepts in coding theory.
In~\cite{BEV13} there is a comprehensive discussion of
$q$-analog of linear codes in the projective space.

\begin{defn} \label{def:quasilinear}
Let $\cU$ be a subset of $\cP_2(n)$ with $\{ {\bf 0} \} \in\cU$. We say
that $\cU$ is a \emph{linear code} in $\cP_2(n)$ if there
exists a function $\linadd: \cU \times \cU \to \cU$ such that
$(\cU,\linadd)$ is an abelian group with the following properties:
the identity element is $\{ {\bf 0} \}$, and the inverse of every
group element $X \in \cU$ is $X$ itself. The function $\linadd$ is
\emph{isometric}, namely:
$$
d_S(X \linadd Y_1,X \linadd Y_2) = d_S(Y_1,Y_2)\quad\text{for all}\quad X,Y_1,Y_2\in \cU~.
$$
\end{defn}

The discussion in~\cite{BEV13} yields a linear code of size $2^n$ in $\cP_2(n)$
and two intriguing questions remained open.

\begin{problem}
Prove that the size of the largest linear code in $\cP_2(n)$ is $2^n$
or show a linear code in $\cP_2(n)$ of size $2^{n+1}$.
\end{problem}

\begin{problem}
Given a linear code $\C$, in $\cP_2(n)$,
which contains $\F_2^n$ as a codeword, prove or disprove that
the number of codewords with dimension $k$ in $\C$ is at most
${n\choose k}$.
\end{problem}

\begin{center}
{\bf Acknowledgments}
\end{center}
I am grateful to Alexander Vardy for many discussions
on the topics mentioned in this paper.
This survey was  motivated by
COST Action IC1104 "Random Network Coding and Designs over GF($q$)".
It should be acknowledged that this action has made it possible to
bring together many researchers in related areas for common discussions
and new joint research on problems related to this paper.


\end{document}